\documentclass[times,a4paper]{article}
\pdfoutput=1

\usepackage{lineno}
\usepackage[colorlinks,bookmarksopen,bookmarksnumbered,citecolor=blue,urlcolor=red]{hyperref}
\modulolinenumbers[5]
\usepackage{balance}
\usepackage{subfig}
\usepackage[utf8]{inputenc}
\usepackage{amsmath}
\usepackage{amsfonts}
\usepackage{textcomp}
\usepackage{graphicx}
\usepackage{marvosym}
\usepackage{comment}
\usepackage{wrapfig}
\usepackage[algo2e, linesnumbered,ruled,vlined]{algorithm2e}
\usepackage{authblk}
\usepackage{fullpage}

\newcommand{\sqb}{\emph{Baseline Uniform Grid}\;}
\newcommand{\sq}{\emph{Uniform Grid}}
\newcommand{\qq}{\emph{Quadtree}}
\newcommand{\SQB}{\textsf{UG$_{Baseline}$}}
\newcommand{\SQ}{\textsf{UG}}
\newcommand{\QQ}{\textsf{QUAD}}
\newcommand{\CPU}{\textsf{CPU-ST}}
\newcommand{\lstparallel}[1]{\mathbf{parallel_{#1}}}
\newcommand{\hide}[1]{}

\providecommand{\DontPrintSemicolon}{\dontprintsemicolon}

\usepackage[normalem]{ulem}
\usepackage{color}
\begin{document}

\newtheorem{proof}{Proof}
\newtheorem{lemma}{Lemma}
\newtheorem{property}{Property}
\newtheorem{definition}{Definition}
\setcounter{secnumdepth}{4}



\title{Manycore processing of repeated range queries\\over massive moving objects observations}
\author[1]{Francesco Lettich \thanks{lettich@dais.unive.it}}
\author[1]{Salvatore Orlando \thanks{orlando@unive.it}}
\author[1]{Claudio Silvestri \thanks{silvestri@unive.it}}
\author[2]{Christian Jensen \thanks{csj@cs.aau.dk}}
\affil[1]{Dipartimento di Scienze Ambientali, Informatica e Statistica, Università Ca' Foscari, Via Torino 155, Venezia, Italy}
\affil[2]{Department of Computer Science, Aalborg University, Selma Lagerlöfs Vej 300, DK-9220 Aalborg Ø, Denmark}


\maketitle

\begin{abstract} 

The ability to timely process significant amounts of continuously updated spatial data is mandatory for an increasing
number of applications. Parallelism enables such applications to face this data-intensive challenge and allows the
devised systems to feature low latency and high scalability.
In this paper we focus on a specific data-intensive problem, concerning the repeated processing of huge amounts
of range queries over massive sets of moving objects, where the spatial extents of queries and objects are continuously
modified over time.
To tackle this problem and significantly accelerate query processing we devise a hybrid CPU/GPU pipeline that
compresses data output and save query processing work. The devised system relies on an ad-hoc spatial index leading
to a problem decomposition that results in a set of independent data-parallel tasks. The index is based on a point-region
quadtree space decomposition and allows to tackle effectively a broad range of spatial object distributions, even those
very skewed. 
Also, to deal with the architectural peculiarities and limitations of the GPUs, we adopt non-trivial GPU data
structures that avoid the need of locked memory accesses and favour coalesced memory accesses, thus enhancing the
overall memory throughput.
To the best of our knowledge this is the first work that exploits GPUs to efficiently solve repeated range queries
over massive sets of continuously moving objects, characterized by highly skewed spatial distributions. In comparison
with state-of-the-art CPU-based implementations, our method highlights significant speedups in the order of 14x-20x,
depending on the datasets, even when considering very cheap GPUs.

\end{abstract}

\section{Introduction}
\label{sec:intro}		

An increasing number of applications need to process massive
spatial workloads. Specifically, we consider applications in settings
where spatial data is continuously produced over time and needs to be processed
rapidly, e.g., scenarios involving mobile device infrastructures, Massively Multiplayer Online Games (MMOG), behavioural simulations where the behaviours of agents
may affect other agents within a given range, and so on.
In these applications, very large populations of continuously \emph{moving
objects} frequently update their positions and issue \emph{range queries} in order to 
look for other objects within their interaction area. The resulting massive workloads pose
new challenges to data management techniques.

\begin{figure}[h]
\centering
    \includegraphics[width=.5\textwidth]{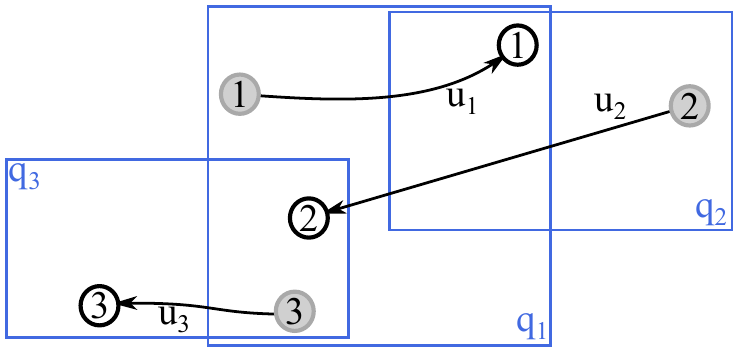}
    \caption{Moving objects, issued queries \textit{q} and position updates \textit{u}.}
    \label{fig:scenario}
\end{figure}

Figure \ref{fig:scenario} shows an instance of this setting with three
objects, namely $o_1$, $o_2$, and $o_3$. Object positions are represented as circles with object identifiers inside, 
while updates are depicted as arrows (labeled $u_1,u_2,u_3$) connecting previous
positions (represented by gray circles) with current positions. Range queries are
shown as rectangles and labeled as $q_1,q_2,q_3$. The result set of query
$q_3$, when executed after $u_2$, is $\{o_2, o_3\}$.

To enable parallel processing and optimizations, and thus manage
the targeted workloads in a scalable manner, we discretize the time in
intervals (\textit{ticks}), assign location updates and queries to the ticks in
which they occur, and process the updates and queries in the resulting
batches such that the query results are reported after the end of each tick.
This approach has the effect of replacing the processing of a large
number of independent and asynchronous queries with the \emph{iterated processing
of spatial joins} between the last  
known positions of all the moving objects at the end of a
tick and the queries issued during the tick. In other words we trade (slightly) delayed processing of queries for increased throughput, and therefore care is needed to ensure acceptable delays.
%

To achieve high performance and scalability we exploit a platform encompassing an
off-the-shelf general-purpose microprocessor (CPU) coupled with a
\textit{Graphics Processing Unit} (GPU), which is a highly
parallel \emph{manycore} architecture featuring hundreds of processing cores \cite{Satish2009, lee2010debunking}. To benefit from GPUs exploitation,
limitations and peculiarities of these architectures must be carefully taken into
account. Specifically, individual GPU cores are slower than those
of a typical CPU, while some memory access patterns may cause serious performance degradation due to
contention and serialization of memory accesses. Effective query processing techniques must address
these limitations:  multiple cores must work together to efficiently process queries in parallel for most of the time, and must 
coordinate their activities to ensure high memory bandwidth. 

With these goals in mind, we present the \qq\ method (hereinafter referred as \QQ), whose design allows an easy and elegant parallelization.
The key idea behind the method is to partition the problem space using a \emph{point-region quadtree} inducing a regular grid; in turn, the cells of the grid represent a set of independently solvable problems, each one associated with a \emph{data-parallel task} runnable on a GPU streaming multiprocessor, also capturing and adapting to possibly skewed spatial distributions.
This strategy allows us to obtain a quasi uniform distribution of the workloads among coarse-grained tasks -- each task corresponding to a single cell of the index -- with the aim of improving the  overall efficiency of the system and maximizing the performance. In order to demonstrate the importance of this aspect we also introduce a baseline spatial index, whose space decomposition relies on the usage of a \emph{uniform grid}. We call this method \sq\ (hereinafter referred as \SQ).

Both \QQ\ and \SQ\ preprocess in parallel the data and store consecutively object and queries
falling in the same grid cell, thus optimizing memory accesses. Further, to avoid the use of blocking writes and to ensure high throughput, both methods compute the query results by means of a two phase-approach using a particular \textit{bitmap} intermediate representation previously introduced by the authors \cite{lettich2014gpu}.


To the best of our knowledge this is the first work that exploits the GPUs
to efficiently solve repeated range queries over massive
moving objects observations, having care to tackle effectively skewed spatial distributions as well.
There are few existing works that use the GPUs for spatial query
processing, but they consider substantially different problems, as detailed in
Section \ref{sec:relwork}.

In a previous work \cite{lettich2014gpu} we already introduced the general framework for repeatedly computing on GPU sets of range queries on streams of objects, but we limited ourselves to unrealistic uniform spatial distributions and to a spatial index relying on a simple uniform grid whose coarseness directly depends on the size of the queries. In this sense, this work aims to extend and improve the previous work by introducing a cleverer space index and a novel set of optimizations.

The main contributions of this work can be summarized as follows: 
\begin{itemize}
\item we define a hybrid CPU-GPU pipeline to process batches of range queries, which effectively exploit the GPU computational power while taking care of its architectural features and limitations. Thanks to its flexibility,  the pipeline can be adapted to different spatial indices.
\item we introduce a set data structures which allow the pipeline to perform operations that concurrently write interlaced lists of results, using coalesced memory accesses that avoid race conditions.

\item while we adopt the usual query splitting approach to make the data-parallel tasks induced by the indices completely independent, we take advantage of subquery areas that \textit{completely cover} index cells to \textit{save work} during the query processing (in terms of amounts of containment tests performed), and \textit{compress} the information the GPU has to send back to the CPU when notifying the query results.  This
optimization is particularly important for facing one of the main data management challenges of our problem, 
related to the size of the aggregate output returned at each tick: for example, when the percentage  of  objects that issue  a  range query is $100\%$, this size is quadratic in the number of moving objects~\cite{lettich2014gpu}.

\item we carry out an extensive set of experiments to compare \QQ\ with \SQ. The structural regularity and simplicity characterizing the uniform grids, onto which \SQ\ relies, is a well-known feature that fits very well the GPUs characteristics, but have structural limitations - mainly related to the inability to fully capture the skewness possibly characterizing the spatial data - which may seriously hinder the performances.

On the other hand, \QQ\ is able to automatically adapt to very skewed spatial object distributions, assuring very good performances for a broad range of spatial distributions. We demonstrate this claim by comparing \QQ\ with \SQ. We also compare \QQ\ with the state-of-the-art (for what relates to the problem considered) sequential CPU algorithm, namely the \textit{Synchronous Traversal} algorithm \cite{sowell2013experimental,brinkhoff1993efficient}.
%
\end{itemize}

The paper is structured as follows: in Section \ref{sec:preliminaries} we proceed describing the problem setting and state the problem addressed. In Section \ref{sec:gpuOverview} we cover the architectural specifics of the GPUs as well as the general constraints to be considered when designing hybrid CPU/GPU processing pipelines coupled with proper data structures. In Section \ref{sec:computeIteratedJoins} we present the spatial indices used by \QQ\ and \SQ\ to partition the workload, while in Section \ref{sec:commonParts} we extensively detail the pipeline as well as the customizations needed by \QQ\ and \SQ. Sections \ref{sec:setup} and \ref{sec:experiments} present an extensive set of experimental studies which show (i) the benefits coming from the usage of the proposed range query processing pipeline and data-structures, (ii) how \QQ\ spatial indexing is better with respect to the one used by \SQ\ and (iii) how \QQ\ outperforms the state-of-the-art sequential CPU competitor as well as outperform, or being on par with, \SQ. Finally, in Section \ref{sec:relwork} we cover the related work while Section \ref{sec:conclusions} gives the conclusions.

\section{Problem setting and statement}
\label{sec:preliminaries}

In this section we provide definitions that capture the problem setting and
the problem we aim to solve.

\subsection{Problem setting}
\label{sec:scenarioDef}

We consider a set of points $O=\{o_1, \ldots, o_n\}$ moving in
two-dimensional Euclidean space $\mathbb{R}^2$, where the position of
object $o_i$ is given by the function $pos_i: \mathbb{R}_{\geq 0}
\rightarrow \mathbb{R}^2$ mapping time instants into spatial positions.

These points model objects that issue position updates and range
queries as they move.  Let ${\cal P}_i = \langle
p_i^{t_0},\ldots,p_i^{t_k},\ldots \rangle$,\ $t_j < t_{j+1}$, be the
time-ordered sequence position updates issued by $o_i$, where
$p_i^{t_j} = pos_i(t_j)$ is a position update.  A range query issued
by object $o_i$ at time $t$ is denoted by $q_i^t = (x^a, x^b, y^a,
y^b)$, where $(x^a, y^a)$ and $(x^b, y^b)$ are the lower left and
upper right corners of a rectangle. Thus, ${\cal Q}_i = \langle
q_i^{t_0},\ldots,q_i^{t_k},\ldots \rangle$, $t_j < t_{j+1}$, is the
time-ordered sequence of queries issued by $o_i$.

Given the above, the most recently known position of $o_i$ before time
$t$, $t\geq t_0$, is denoted as $\hat{p}_i^t$ and defined as follows:
\[\hat{p}_i^t = p_i^{t_k} \in {\cal P}_i \text{ if } t_{k} < t \le t_{k+1}\]
Similarly, the most recent query issued by $o_i$ before time $t$, $t\geq
t_0$, is $\hat{q}_i^t$:
\[\hat{q}_i^t = q_i^{t_k} \in {\cal Q}_i \text{ if } t_{k} < t \le t_{k+1}\]

We assume that the processing of a query can be delayed to a certain
extent in order to optimize the overall system throughput. We process
queries using the most up-to-date information available.  

\begin{definition}[\textsf{Result set of a range query}]
\label{def:rangeQuery}
The result of query $q_i^t$ when computed at time $t'$, $t_0 \le t \le t'$, is
denoted by $res(q_i^t, t')$ and is defined as follows.
$$res(q_i^t, t') = \{ o_j \in O \ |\  \hat{p}_j^{t'} \in_s q_i^t\}, $$
where $\hat{p}_j^{t'} \in_s q_i^t$ denotes that $\hat{p}_j^{t'} =
(x,y)$ belongs to the query rectangle $q_i^t$, i.e., $x^a \le x \le
x^b$ and $y^a \le y \le y^b$.
\end{definition}

Assuming that updates $u_1, \dots, u_3$ in Figure~\ref{fig:scenario}
are the most recent ones before $t'$, we have $res(q_3^t,t') = \{o_2,
o_3\}$, which includes also object $o_3$ that issued the query.

\subsection{Batch processing}
 
To obtain high throughput when facing massive workloads due to
frequent updates and queries issued by huge populations of
moving objects, we quantize time into \emph{ticks} (time intervals) with the objective
of processing updates and queries in batches on a per-tick basis.
Assuming that the initial time is 0 and the tick duration is $\Delta
t$, the $k$-th time tick $\tau_k$ is the time interval $[k \cdot
\Delta t,\ \ (k+1) \cdot \Delta t)$.
Specifically, we aim to collect object position updates and queries that arrive
during a tick, and process them at the end
of the tick. 
%
If an object submits more than one update
and query during a time tick, only the most recent ones are
processed. 

Let $P^{\tau_k}=\{\hat{p}_1^{(k+1) \cdot \Delta t}, \ldots,
\hat{p}_n^{(k+1) \cdot \Delta t}\}$ be the last known positions of all 
objects at the beginning of the $(k+1)$-th tick, and 
%
$Q^{\tau_k}=\{q_1^{\tau_k}, \ldots, q_n^{\tau_k}\}$ be
the most recent queries issued during the $k$-th tick, where:

\[ q_i^{\tau_k} = \left\{ 
  \begin{array}{l l}
    \hat{q}_i^{(k+1) \cdot \Delta t} & \quad \text{\small If object $o_i$
issues any query during the $k$-th tick.}\\[.2cm]
    \bot & \quad \text{\small Otherwise.}
  \end{array} \right.\]

\noindent Note that if object $o_i$
does not issue any query during the tick, then $q_i^{\tau_k}=\bot$.

Figure~\ref{fig:timeline} captures the temporal aspects of the
previous example. The timeline is partitioned into ticks $\tau_1$, $\tau_2$, $\ldots$ of duration
$\Delta t$.
\begin{figure}[h!]
        \centering
                \includegraphics[width=.42\columnwidth]{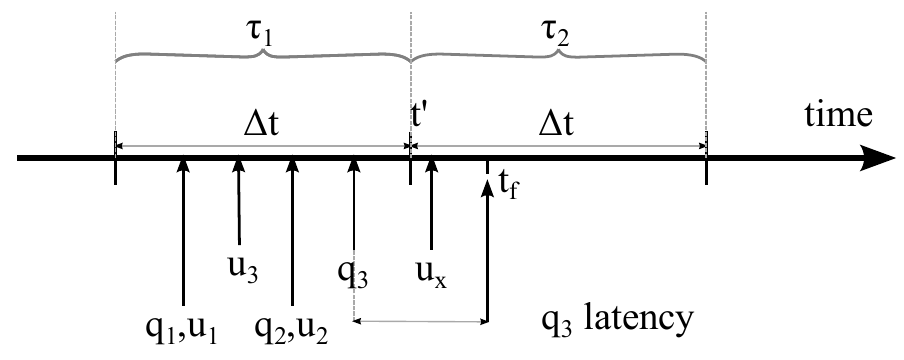}
                \caption{Timeline}
                \label{fig:timeline}
\end{figure}
Objects issue updates and queries independently and
asynchronously. For example object $o_3$ sends a query and an update
separately. Incoming updates and queries are batched based on the
ticks. For example, update $u_x$ belongs to the batch of $\tau_2$. The
batches are processed at the beginning of the next tick. Thus, at time $t'$ (at the beginning of
$\tau_2$) we start processing all updates and
queries arrived during $\tau_1$. We complete the processing of the batch, thus making available the query results, 
at time $t_f$, hopefully before the end of $\tau_2$.

\subsection{Query semantics}

The procedure for computing queries as described above ensures
serializable query processing and implements the timeslice query
semantics, where query results are consistent with the
database state at a given time, usually when we start processing the query. 
This is a popular choice in traditional databases
and is commonly adopted in literature~\cite{dittrich2011movies,gray1993transaction,vsidlauskas2012parallel,kornacker1997concurrency}.



For example, in Figure~\ref{fig:timeline}, 
the computation results are returned at time $t_f$ for the first batch. Since the object update $u_x$ occurs at time $t_u$, 
$t_u > t'$,
i.e., after we
start processing $q_3$, $u_x$ is
not considered even if it arrives before the results are returned. 
In this way,
the query result is consistent with the database state at time $t'$, when we start processing query $q_3$.

\subsection{Quality of Service - Query latency}
\label{sec: qos considerations}

On the one hand, the processing of updates and queries on large batches
can be expected to improve system throughput. On the other hand, we
assume that some applications, e.g., MMOG applications, are sensitive
to the delays with which query results are returned.
Thus, it is important to be able to assess the \emph{latency} of
query processing, 
which is affected by the number of queries and
updates that arrive during a tick, the tick duration, and the
computational capabilities of the system.

\begin{definition}[\textsf{Latency, Queueing, Computation time}]
\label{def:latency}
Assume that the processing of query $q_i^t$, issued at time $t$,  starts at time $t'$ and
completes at time $t_f$. We define the following durations: 
$\mathit{latency\_time}(q_i^t) = t_f - t$, $\mathit{queueing\_time}(q_i^t) = t' -
t$, and $\mathit{processing\_time}(q_i^t) = t_f - t'$, so that 
$\mathit{latency\_time}(q_i^t) = \mathit{queueing\_time}(q_i^t)  + \mathit{processing\_time}(q_i^t)$.
\end{definition}


%

We can now generalize the concept of latency to all the queries issued during
a tick, all executed in batch at the beginning of the next tick.

\begin{definition}[\textsf{Tick latency}]
\label{def:latency-tick}
The latency of the queries $Q^{\tau_k}$ arrived during the
$k$-th tick is defined as follows:
$$\mathit{Tick\_Latency}(Q^{\tau_k}) = \max_{i\in\{1, \ldots, n\},\ \ q_i^{\tau_k} \neq \bot} \mathit{latency\_time}(q_i^{\tau_k})$$
\end{definition}

Given an application-dependent maximum latency threshold $\lambda$, a
system satisfies the application's \emph{QoS Latency Requirement} if,
for each tick $k$, $\mathit{Tick\_Latency}(Q^{\tau_k})  \leq \lambda$.

The tick duration $\Delta t$ should be chosen such that even queries
issued at the beginning of a tick are answered within time duration
$\lambda$. 
Since query processing is delayed till the beginning of the next tick,
the \emph{worst-case latency} for a query $q_i^t$ takes place when $q_i^t$ 
is issued at the beginning of a tick:
in this case, the latency is the sum of $\Delta t = \mathit{queueing\_time}(q_i^t)$ and 
$\mathit{processing\_time}(q_i^t)$. The following lemma states a simple, sufficient
criteria to select $\Delta t$ or to verify whether a given execution
time satisfies the latency requirement.

\begin{lemma}
\label{lemma:timely}
Let $\Delta t^k_{exe}$ be the time to process all queries in the $k$-th
batch. Given a tick duration $\Delta t$ and a latency requirement
$\lambda$,  then if $\Delta t + \Delta t^k_{exe} \le \lambda$, the execution
satisfies the latency requirements, i.e.,
$\textit{Tick\_Latency}(Q^{\tau_k}) \le \lambda$.

From the above we can derive the following sufficient
condition for $\textit{Tick\_Latency}(Q^{\tau_k}) \le \lambda$:
\begin{equation}
\label{eq:timely}
\Delta t > \lambda - \Delta t \ge \Delta t^k_{exe}.
\end{equation}
\end{lemma}

Above, we have that $\Delta t > \lambda - \Delta t$ because the
processing of the queries accumulated during a tick is assumed to be
completed before the end of the next tick.

The computational capabilities of a system influences the choice of
the tick duration and the fulfillment of the latency requirement in
Lemma~\ref{lemma:timely}.

\begin{lemma}
\label{lemma:timely2}
Let $\beta$ be the system bandwidth, expressed in terms of the number of queries processed per time unit, 
and let $Q_{max}$ be the maximum number of
queries that can occur during a tick. Then a sufficient condition for
the system to meet the QoS latency requirement (based on threshold  $\lambda$) is:
%
\begin{equation}
\label{eq:timely2}
\beta \ge  \frac{Q_{max}}{\lambda - \Delta t}
\end{equation}
%
\end{lemma}

\begin{proof}
According to Equation~\ref{eq:timely}, in order to respect the timeliness,
we have to process all queries in a time $\Delta t^k_{exe}$ such that  $\Delta t^k_{exe} \leq 
\lambda - \Delta t < \Delta t$. Given the bandwidth $\beta$, 
the maximum execution time to process all queries of a tick is $\frac{Q_{max}}{\beta}$. Hence, 
$\frac{Q_{max}}{\beta} \leq \lambda - \Delta t$ must hold, from which the lemma follows.
\end{proof}

For a given latency requirement $\lambda$, if we increase the tick
duration $\Delta t$, this increases $Q_{max}$ and decreases $\lambda -
\Delta t$. So, if we increase $\Delta t$, in order to
satisfy Equation \ref{eq:timely2} we have to compute more queries in
less time, and thus it may happen that bandwidth $\beta$  becomes insufficient
to support the requested workload respecting the given latency threshold, i.e., $\beta <  \frac{Q_{max}}{\lambda - \Delta t}$. 


\subsection{Problem statement}
\label{sec:problem}


We can finally state the problem of computing repeated range queries over massive streams of moving objects observations, by discretizing the time in intervals (ticks), synchronizing query processing according to these ticks, and iteratively computing 
queries in batch mode.
 
Given \emph{(i)} a set of $n$ objects $O$, \emph{(ii)} a partitioning of the time
domain into ticks $[\tau_k]_{k \in \mathbb{N}}$ of duration $\Delta
t$, \emph{(iii)} a query latency requirement $\lambda$, and \emph{(iv)} a sequence
of pairs $[(P_{\tau_k},Q_{\tau_k})]_{k \in \mathbb{N}}$, where
$P_{\tau_k}$ is the up-to-date object positions at the end of
$\tau_k$, and $Q_{\tau_k}$ is the set of the last issued queries
during $\tau_k$, 
we have that the \textbf{iterated batch processing} of queries $Q_{\tau_k}$ over the corresponding 
$P_{\tau_k}$, $k \in
  \mathbb{N}$, \textbf{yields} $[R_{\tau_k}]_{k \in \mathbb{N}}$, i.e., a \textbf{sequence of pairs}, each composed of a query and the list of the corresponding results:
\[R_{\tau_k} = \{
(q_i^{\tau_k}, \ res(q_i^{\tau_k}, \ (k+1) \cdot \Delta t)) \ \ | \ \
q_i^{\tau_k} \neq \bot \ \wedge \ q_i^{\tau_k} \in
Q_{\tau_k} \}.\]

The \textbf{processing time} of each batch of queries $Q_{\tau_k}$ 
must be \textbf{upper
boun\-ded} as follows, to satisfy the query latency requirement
$\lambda$:
\[\mathit{processing\_time}(Q_{\tau_k}) \le
(\lambda - \Delta t) < \Delta t\]
%



\section{Graphics Processing Units}
\label{sec:gpuOverview}

GPUs are based on massively parallel computing architectures that
feature thousands of \textit{cores} grouped in \textit{streaming
multiprocessors}
\footnote{For the purposes of this work we will use the NVIDIA CUDA terminology to refer to the GPUs architectural features and peculiarities, as well as to describe software targeted to GPUs, since CUDA represents the dominant framework in the context of general purpose computing on GPUs.}
(hereinafter denoted by SMs for brevity) coupled with several gigabytes of high-bandwidth RAM. 
During the latest years these devices sparked a consistent interest due to their ability in performing
general purpose computations which, together with the amount of available cores
in their architectures, offer a potential for substantial
performance gains when compared to the performance of traditional
CPUs.

Due to the architecture of these devices, and depending on the problem considered, effectively exploiting their computational power is usually far from trivial. 
Specifically, each GPU processing core is slower than a typical CPU and has limitations on
its access to device memory, resulting in potential contentions unless
specific conditions are satisfied~\cite{hong2009analytical}. Moreover, GPU cores have to 
coordinate their actions, which is usually a complex issue considered their architectural
organization.

Proper algorithms, designed with the architectures
of the GPUs in mind, are needed in order to maximize the performances and
obtain significant gains with respect to CPU-based algorithms, a goal which
is not always possible to pursue \cite{lee2010debunking} depending on the
characteristics of the targeted problem.

In order to exploit effectively the computational power of a GPU, memory accesses 
should generally have high spatial locality in order to exploit GPUs demand-fetched caches as much as possible \cite{jia2012characterizing}. In addition, we have to ensure, whenever possible, that all the cores of an SM profit from memory block transfers, by forcing coalescing of parallel data transfers: this avoids serial memory accesses and consequent performance degradation due to 
a sub-optimal usage of memory hierarchies.
%

Moreover, GPUs feature several types of memories ranging from private 
thread registers and fast shared memory, which are both shared among the core groups
of each SM, to global memory, which has a lower throughput but it is of significant
size and represents the contact point with the CPU host. 
To achieve consistent performances 
a programmer has to be aware of this complex memory hierarchy by 
orchestrating and managing explicitly memory transfers
between different memories.

Finally, workload partitioning is paramount when designing GPU
algorithms since unbalances may create inactivity bubbles
across the streaming multiprocessors and seriously cripple the
performance.

A GPU consists of an array of $n_{SM}$ multithreaded SMs,
each with $n_{core}$ cores, yielding a total number of $n_{SM}\cdot
n_{core}$ cores. Each SM is able to run \emph{blocks} of
\emph{threads}, namely \emph{data-parallel tasks}, with the threads 
in a block running concurrently on the cores of an SM.  Since a block typically has many more threads
than the cores available in a single SM, only a subset of threads, called
\emph{warp}, can run in parallel at a given time instant. Each warp consists of $sz_{warp}$ \emph{synchronous, data parallel 
threads}, executed by an SM according to a SIMD paradigm \cite{hong2009analytical,patterson2012comparch}.
Due to this behavior, it is important to avoid branching inside the
same block of threads. It is worth remarking that at warp level no synchronization mechanisms are 
needed to guarantee data dependencies among threads, thanks to the underlying scheduling. 
%
Finally, a function designed to be executed on GPU is called \textit{kernel}.

\subsection{Main algorithmic design issues}
\label{sec:design issues}

Considering the specificities of the problem described in Section \ref{sec:problem}, five main design issues shall drive the design of the hybrid CPU/GPU pipeline in charge of the query processing. First, we have  to find a proper way to distribute the workload evenly among the GPU streaming multiprocessors, since unbalances typically create inactivity bubbles. Second, 
we need to avoid contention/serialization when accessing the GPU device memory, in order to favour spatial locality, thus 
properly taking advantage from the complex GPU memory hierarcies. Third, we should
compress the data the GPU has to send back to the CPU during the query processing, since in our problem 
the output is typically much larger than the input. Fourth, \emph{(iv)} expensive synchronization mechanisms among concurrent threads should be avoided, since these are typically very costly in terms of performance. Finally, for each pipeline task executed on the GPU  the unit of parallelization (either objects or queries, or partitions of queries) should be carefully chosen according to the task specificities.

\section{Spatial indexing and data structures}
\label{sec:computeIteratedJoins}


In this section we discuss the inspiring principles behind \QQ\ and \SQ, 
along with the architectural features of GPUs that impact on the spatial indices and data structures design. 


When processing repeated range queries, the same procedure is repeated for each tick. Thus, 
for the sake of readability, hereinafter we omit the subscript that indicates the tick, and denote by  
$P$, $Q$, and $R$, respectively, the up-to-date object positions, the non-obsolete queries, and the result set associated 
with a generic tick.

\subsection{Design considerations}

A brute-force approach for computing repeated range queries entails $O(|P|\cdot|Q|)$ containment checks per tick. By using spatial indices it is possible to prune out consistent amounts of pairs of queries and object locations that do not intersect.
However, when choosing or designing an appropriate index,
we have to consider its pruning power along with its maintenance costs.
For example, regular grid indices are generally reported to have low maintenance costs, and thus are suitable for update-intensive settings~\cite{sid11}.
Another aspect is the number of cores and the memory hierarchy
provided by the underlying computing platform. Given the same workload, different indices may
be the best option for different platforms.
With massively parallel platforms such as  GPUs, the regularity
characterizing spatial indices based on regular grids is attractive,
as it enables fast and efficient parallel index updating and querying. Even if 
tree-based spatial indices are able to distribute objects evenly 
among the index cells (the tree leaves), 
we have to avoid navigating the tree,  
since this may severely hinder efficiency due to poor data locality when accessing the memory.

In a previous work \cite{lettich2014gpu} we devised a simple uniform grid-based spatial indexing used to partition the workload and prune out useless containment tests. In that work we chose to determine the size of the index grid cells on the basis of the 
query size: the rationale was to reduce the amount of index cells to be considered when processing each query.   
However, solutions based on uniform grids generally cannot cope efficiently with skewed spatial distributions. 
To solve this issue, in this paper we propose the \QQ{} method, which relies on a tree-based recursive spatial indexing, induced by point-region quadtrees.
%
To ensure an unbiased comparison between \QQ\ and uniform grid-based spatial indices, 
we also introduce the \SQ\ method as a baseline.
\SQ\ relies on a simple uniform grid-based spatial index, without any a-priori constraint on the size of the grid cells.
Also,  \SQ\  integrates all the optimizations conceived for \QQ\ (such optimizations are detailed in Section \ref{sec:optimizations}).


\subsection{Overview of the methods}
\label{sec:methods overview}

In the following we give an overview of both \SQ\ and \QQ.

\SQ\ materializes at each tick a uniform grid over the minimum bounding rectangle enclosing the object positions. 
The only way \SQ\ can cope with data skewness is by changing the coarseness of the grid, targeting a coarseness tradeoff on the basis of the object densities in crowded areas and loosely populated ones.

\QQ\ still yields at each tick an index over the same bounding rectangle, but the index cells are of varied size as 
\QQ\ is able to dynamically tune their size according to local object densities. 
To this end \QQ\ utilizes a point-region quadtree, which entails a space partitioning that ensures a pretty balanced distributions of objects among the index cells even in presence of skewed data. 
Even if tree data structures are, in principle, difficult to manage on GPUs, 
the direct relationship between the quadtrees structural properties and the Morton codes~\cite[Ch.\ 2]{har2011geometric}\cite{raman2008converting}
 open up to the possibility of implementing efficient massively parallel quadtree construction and lookup algorithms on GPUs\cite{lauterbach2009fast}.

Regardless of the index adopted, 
queries can be processed concurrently according to a \emph{per-query parallelization}.
%
%
More specifically, since both indices induce a partition of the space, where each disjoint space tile corresponds to a cell 
of either \SQ\ or \QQ, we virtually split queries according to the space partition, thus producing a \emph{subquery} for each index cell a query intersects.
Indeed, each subquery yields an \emph{independent subtask}, which we process in parallel by only accessing the 
objects falling in the associated index cell. Note that this approach 
decreases the overall amount of containment checks, although the splitting yields more subtasks to process as 
the same query is processed several times, once for each relevant index cell. 

It is worth pointing out that we can have subqueries whose areas entirely cover small index cells.
This allows us to strongly optimize the computation: first, we can compress the output, since all the objects of 
these cells falls into the subquery areas; second, for the same reason we can avoid processing 
covering subqueries, thus saving computation time (see the \textit{covering} queries optimization, Section \ref{sec:optimizations}).


Both \SQ\ and \QQ\ use an ad-hoc lock-free data structure based on bitmaps~\cite{lettich2014gpu}, to manage the result sets while they are produced on GPU. This design choice entails a further post-processing step needed to enumerate the final results contained in this data structure.
To prove the merit of this choice we consider a more basic baseline, i.e., 
a variant of \SQ, namely the \sqb\ (\SQB) method, which uses atomic operations to ensure the consistency of the result sets content. 

\subsection{Space partitioning and indexing}
\label{sec:partitioningOverview}

In the context of parallel query processing, there are two main reasons for partitioning and indexing the data according to a given space partitioning approach: the first one, also common to sequential query processing, is to avoid redundant computations and access to irrelevant data, while ensuring fast access to relevant information. The second reason, which is triggered by the ability to process data in parallel, 
is to ensure independent computations, avoid redundant work and balance the workload among the processing unit cores.

In the following we introduce the two space partitioning methods onto which \SQ\ and \QQ\ rely. 
Both methods aim to adaptively partition the Minimum Bounding Rectangle (\textsf{MBR}) containing all the object positions during any tick. We denote this \textsf{MBR} by ${\cal G}=(x^{\cal G}_a,y^{\cal G}_a,x^{\cal G}_b,y^{\cal G}_b)$, where 
$(x^{\cal G}_a,y^{\cal G}_a)$ and $(x^{\cal G}_b,y^{\cal G}_b)$ represent the lower-left and the upper-right corners of the MBR.
Aside from the specific ways through which the methods define the geometries and enumerate grid cells, both of them assign queries and object locations according to the same mapping functions (see Section \ref{sec:mapping functions}).

\subsubsection{Uniform grid-based partitioning}
\label{sec:gridOverview}
\SQ\ partitions the space by superimposing a uniform grid ${\cal C}$, whose cells are of equal size, 
over ${\cal G}$.

\begin{definition}[\textsf{MBR partitioning into a uniform grid}]
  \label{def:MBRPartitioning}
\textit{${\cal G}$ is partitioned according to a uniform grid ${\cal C}$ of $N \cdot
  M$ cells of width $W$ and height $H$ such that the cell $c_{ij}$ covers
  the following region:
$$(x^{\cal G}_a + i \cdot W,\ \ \ x^{\cal G}_a + (i+1) \cdot W,\ \ \
y^{\cal G}_a + j \cdot H,\ \ \ y^{\cal G}_a + (j+1) \cdot H).$$
\noindent 
To ensure that the grid covers ${\cal G}$, constants $N$, $M$, $W$,
and $H$ are chosen so that $x^{\cal G}_a + N \cdot W\ \geq\ x^{\cal
  G}_b$ and $y^{\cal G}_a + M \cdot H\ \geq\ y^{\cal G}_b$ hold.
We associate with each cell $c \in \mathcal{C}$ an integer ID, which enforces a total order among the index cells, by preserving spatial locality.}
\end{definition}

\subsubsection{Quadtree-based partitioning}
\label{sec:quadPartitioningOverview}
In the \emph{quadtree based partitioning} case, ${\cal G}$ is covered by a quadtree-induced grid ${\cal C}$, determined on the basis of the local densities of moving objects. In this case ${\cal G}$ is therefore partitioned into a set of cells 
corresponding to the quadtree leaves.

\begin{definition}[\textsf{MBR partitioning into a quadtree-induced regular grid}]
\label{def:QuadPartitioning}
\textit{${\cal G}$ is partitioned into a set of variably sized
cells belonging to grid ${\cal C}$, induced by a point-region quadtree.
Given a constant $th_{quad}$, denoting the maximum amount of objects allowed inside a single quadrant/cell of the final grid, 
we have that each cell of ${\cal C}$
corresponds to a quadtree leaf, and contains an amount of object not greater than $th_{quad}$.
We associate with each cell $c \in \mathcal{C}$  an integer ID, which enforces a total order among the index cells, by preserving spatial locality.}

\end{definition}



\subsubsection{Mapping of moving objects and queries to space partitions}
\label{sec:mapping functions}

Given an index $\cal C$ derived by \QQ\ or \SQ, we assign objects and queries to the index cells. Since the area of any query can intersect several cells of $\cal C$, this entails a partition of the area.
We call this operation \textbf{query splitting}, which potentially yields a set of \textit{subqueries} for each query.
Finally, each subquery can be univocally assigned to a single index cell.

\begin{definition}[\textsf{Mapping functions for object locations and subqueries}]
\label{def:objPosMap}
\textit{Given the set of cells of a grid ${\cal C}$, we have two \textbf{mapping functions} $f: P \rightarrow {\cal C}$ and $g: Q \rightarrow 2^{\cal C}$ \textbf{map}. Function $f$ maps 
each object location $p \in P$ to the cell $f(p)$ that contains $p$.
Function $g$ maps each query $q \in Q$ to a set of cells $g(q)$, whose intersection with $q$ is not empty. We use the term \textbf{subqueries} to denote the restrictions of a query $q$ to each of these cells. Moreover, we call the operation performed by $g$ \textbf{query splitting}. 
Finally, each subquery is classified as \textbf{intersecting} or \textbf{covering}, according to the fact that it partially/entirely covers the associated cell.}
\end{definition}

\begin{figure}
\centering
    \includegraphics[width=.35\textwidth]{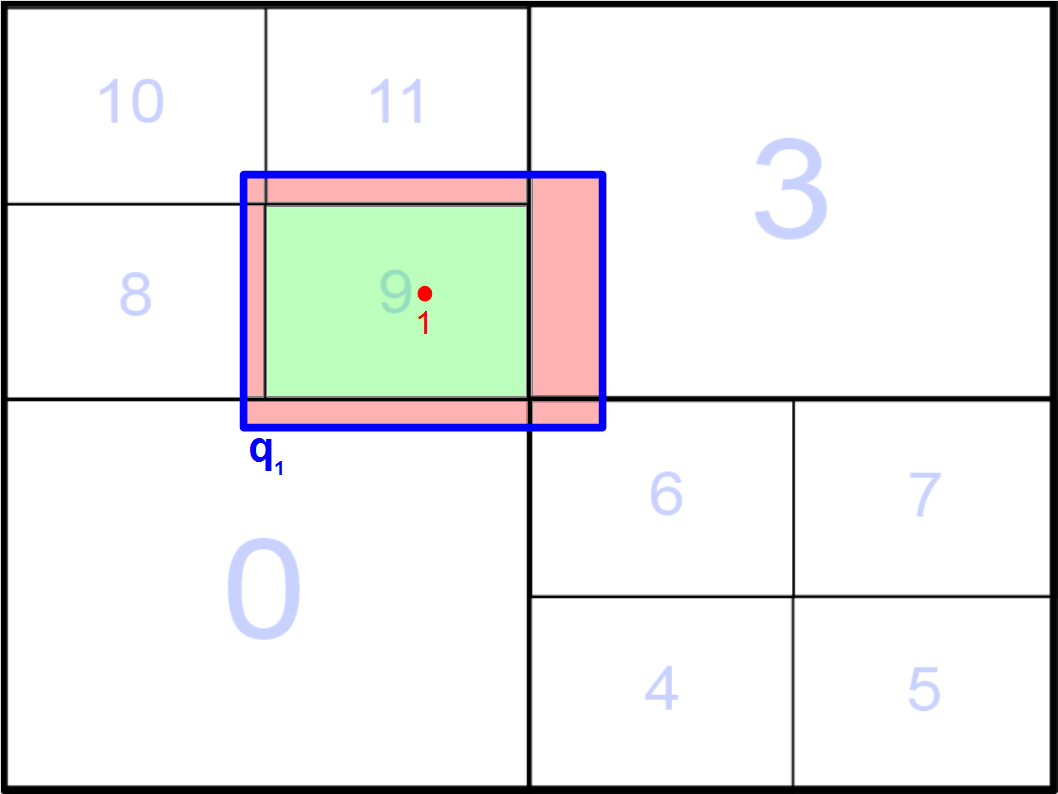}
    \caption{Simple mapping example over a quadtree-induced grid.}
    \label{fig:mappingQueryObjects}
\end{figure}

A simple example about how $f$ and $g$ operate is given in Figure \ref{fig:mappingQueryObjects}: \textit{f} maps object 1 to the cell having ID 9, while \textit{g} splits query $q_1$ (issued by object 1) over 7 different cells. Among these, six are \textit{intersecting} ones (highlighted in pink, namely the subqueries intersecting cells with IDs 0, 3, 6, 8, 10, and 11) while one is a \textit{covering} subquery (highlighted in green, covering the cell with ID 9).
%



\subsection{Data structures}
\label{sec:dataStructuresDesign}

As it will be pointed out in in Section \ref{sec:commonParts}, both \QQ\ and \SQ\ rely on a hybrid CPU/GPU processing pipeline, a pattern quite common in the context of General Purpose Computing on GPUs \cite{sengupta07,merrill11}.
Each stage of the pipeline performs a set of transformations on the data in order to produce a final output. 
To this end, 
%
the design of data structures should (i) allow data to be concurrently accessed  with minimal use of atomic operations or barriers,
thus avoiding locking related penalties; (ii) permit the use of coalesced memory accesses, in order to maximize the memory throughput; (iii) exploit spatial locality, whenever possible, in order to maximize the benefits deriving from coalescing and caching.
In the following subsections we introduce the relevant data structures used by our approach.

\subsubsection{Moving objects and queries data structures and their layout}
\label{sec:dataRep}

\begin{wrapfigure}{r}{0.35\textwidth}
\vspace{-1em}
\begin{center}
    \includegraphics[width=.35\columnwidth]{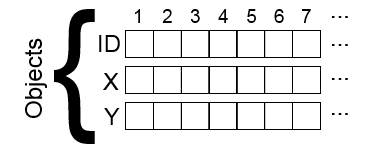}
    \caption{Example of a structure of vectors used on a set of objects described by 3-tuples.}
    \label{fig:structofvec}
\end{center}
\vspace{-1em}
\end{wrapfigure}

Given a set of $n$-tuples representing a class of entities (in our context an object location or a query), the tuples elements are logically arranged by means of a 
\textit{structure of vectors} (also known as \emph{structure of streams} or \emph{structure of arrays}) layout \cite[Ch.33]{Pharr2005}. This layout groups a set of $n$ vectors, each one representing a single element of the tuples, and aligns the vectors elements with respect to the entities they are associated with. An example of such arrangement is given in Figure \ref{fig:structofvec}.

The main benefits of this representation derive from the observation that it does not require complex pointer
arithmetic, and it naturally makes possible to exploit coalesced memory accesses.
Moreover, it is a representation commonly used in well established GPU algorithms, thus allowing an efficient interplay (and code reuse) between the operations making up the processing pipeline.

%
%

%

Consider that we aim at generating \emph{independent tasks}, each one associated with a \emph{subquery} and a specific \emph{active cell}, 
i.e., a cell with at least one object and one subquery, where 
each task processes a single subquery over the objects of the associated cells (see Sections \ref{sec:filterOverview} and \ref{sec:decodeOverview} for more details). 
In light of this, it is convenient to properly arrange the structures of vectors
associated with objects and subqueries (each set has its own structure)
in order to exploit data spatial locality and boost  memory throughput. To this end, we have to arrange entities falling inside the same grid cell in contiguous memory locations (memory blocks). 
Therefore, first object locations and subqueries are sorted by the IDs of the associated index 
cells. 
As a side-note, we observe that the same originating query can be stored in several memory blocks, since function $g$ 
potentially yields multiple intersecting cells for each query. 
Second, block boundaries are stored in a table, by distinguishing between the sub-blocks storing object locations and sub-blocks storing subqueries. This allow us to directly access the data belonging to any cell. The reader may refer to Section \ref{par:indexDetOverviewQuad} 
for more details about the sorting operations performed in order to achieve such arrangements.

%

\subsubsection{Intermediate bitmap representations of query result set}
\label{sec:bitmapRep}

One of the main issues is related to 
the efficient collection of possibly huge sets of query results.
According to the problem statement given in Section \ref{sec:problem}, 
the result of a single tick is described in terms of a set of pairs, each one consisting of an identifier associated with the object issuing a query and a set of identifiers related to the objects falling inside the query result set. 
Since query results are produced concurrently, contentions when writing them out could seriously cripple 
massive parallelism.

To avoid this issue we exploit and improve the two-phase approach we introduced in a previous work \cite{lettich2014gpu}, which relies on two intermediate data structures based on a \textit{bitmap} layout in order
to eliminate the need of threads synchronizing mechanisms while maximizing the overall memory throughput and 
minimize the amount of space used to store intermediate results on GPU. 
Since these data structures are strongly tied to the design of the algorithms in charge of the aforementioned operations, we postpone their description to Sections \ref{sec:filterOverview} and \ref{sec:decodeOverview}.

\section{Query processing pipeline}
\label{sec:commonParts}

The core computation to process each set of range queries can be surely ascribed to the containment tests between objects locations and query areas. Considering the potential huge amount of containment tests and results each tick can yield, this apparently simple and straightforward operation is very expensive.
In order to improve its efficiency we embed this computation in a pipeline of concatenated operations. The various stages of the pipeline prepare the spatial index for improving the efficiency of the containment tests, compute the containment test outcomes in an intermediate format for efficiency reasons, and post-processes these results to produce the final query results.

In the following subsections we introduce the high-level pipeline, common to all methods.
We also discuss the main design differences between \SQ, \SQB\ and \QQ\ in the implementation of each pipeline phase.

\subsection{Pipeline description}
\label{sec:pipelineOverview}
The data entering the processing pipeline at the end of each tick are first processed to select the index parameters and build an empty index (phase 1, \textit{index creation}). Then (phase 2, \textit{moving object and query indexing}), objects and subqueries are mapped to index cells, and finally are sorted so that those contained in the same cell are stored in contiguous memory locations. 
The subsequent phase computes the containment tests between range queries and object locations (phase 3, \textit{filtering with bitmap encoding}) producing an intermediate bit-encoded output which is structured to avoid contentions in memory access (issues \emph{(ii)} and \emph{(iv)} in Section \ref{sec:design issues}). These intermediate results need a final post-processing phase to extract the final results (phase 4, \textit{bitmap decoding}).  
Each phase takes advantage of a tight cooperation between the GPU and the CPU. 

One of the key features of \QQ\ and \SQ\ is the ability to split the computation of each query among the space partitioning elements (cells) it intersects to reduce the total amount of containment tests. This entails the creation of a new set of \textit{subqueries} originating from the query set $Q$.
The distinction between \SQ{} (\SQB{}) and \QQ{} is related to the way they partition the space, that is, how a grid is materialized over the space and how objects locations and subqueries are mapped to grid cells. These aspects involve just the phases 1 and 2 of the pipeline, since the remaining ones directly use the cell identifiers associated with object locations and subqueries to determine which object locations and subqueries are relevant for a specific operation. 

For this reason, in the following we describe phases 1 and 2 separately for \SQ{} (\SQB{}) and \QQ{}, while the remaining ones can be described regardless of the involved spatial index.

\subsection{Index creation and indexing in \SQ{} and \SQB{}.}
\label{par:indexDetOverviewSQ}

%
The performances of this method are significantly affected by the cells size used for $\cal C$. Choosing a suitable value is challenging since it depends on several factors, from the spatial distribution of data, to the opportunity of avoiding part of the computations thanks to optimizations that are triggered locally by grid and query based conditions (e.g., by exploiting the \textit{covering} subqueries optimization described in Section \ref{sec:optimizations}). 

In a previous work \cite{lettich2014gpu} we were able to optimally determine the cell size of a 
uniform grid index, by assuming unrealistic uniform spatial distributions of objects.
%
%
%
%
Since the optimal granularity cannot be decided for any kind of dataset, but we need to still use uniform grid
indexes as baselines for \QQ, we exploit an oracle to choose the grid coarseness for both \SQ{} and \SQB{}. In practice, we determine the optimal grid coarseness parameter, for each tick and any kind of dataset, by performing parameter sweeping, and finally selecting the parameters that are the most favorable to \SQ{} and \SQB{} in each comparison. 
Accordingly, the goal of the index creation phase for the baselines \SQ{} and \SQB{} is simply the ad-hoc choice of the best grid granularity, in order to maximize the performance of the subsequent phases.

\subparagraph*{Index creation (\SQ{}/\SQB{}).}
Since we already know the optimal grid cell size, the goal of the \emph{index creation} phase in \SQ{}/\SQB{} is to determine the minimum rectangle ${\cal G}$ that bounds all the objects (\emph{MBR}). The computation of the MBR is based on a GPU parallel reduction operation over the set of object positions and queries yielding the minimum and maximum coordinates.

Once ${\cal G}$ is set up, we use the cell size determined by the oracle to materialize an optimal uniform grid $\cal C$ over ${\cal G}$, so that objects and queries can be indexed accordingly.

Each cell of ${\cal C}$ is naturally associated with a pair $(i,j)$, identifying the row and the columns of each cell.
However, we adopt a transformation of $(i,j)$ into a uni-dimensional identifier $CellID$, derived from $(i,j)$ by interleaving the binary representations of the two coordinates, thus obtaining the Morton code $z(i,j)$\footnote{To this end, we adopt an optimized bitwise algorithm.}.

\subparagraph*{Moving objects and queries indexing (\SQ{}/\SQB{}).}
%

Given an index $\cal C$,  function $f$ (Definition~\ref{def:objPosMap}) maps a generic object location $p \in P$ to a cell $c \in \cal C$. In \SQ{}/\SQB{} the function consists of a simple algebraic expression that determines grid coordinates (which indeed correspond to a unidimensional Morton code identifying the cell) 
from object locations.
This  is implemented on the GPU by applying  function $f$ in parallel to all elements of $P$, 
thus obtaining a vector whose elements represent cell identifiers corresponding to each object location.

Still on the basis of index $\cal C$,  function $g$ (Definition~\ref{def:objPosMap}) maps a generic range query $q \in Q$ to a set of cells in $\cal C$. The corners of each query $q$ are mapped to grid coordinates, then a nested loop is used to enumerate the identifiers of cells intersected by the query.  
Since containment tests are superfluous for cells completely covered by $q$, the corresponding subqueries are marked as \textit{covering} to enable the optimizations described in Sec.\ref{sec:optimizations}.

In our GPU implementation of $g$, each query \textit{q} is processed by a GPU thread that produces a set of triples \textit{(queryID, cellID, coveringFlag)}\footnote{In practical terms, the \emph{coveringFlag} can be properly embedded inside the integer representing \emph{cellID}.}
, each one representing an intersecting (covering) subquery. To avoid output write contentions without resorting to blocks and synchronization, a two-pass approach is adopted: the first 
dry-run pass determines the amount of triples per query, while the second pass writes out the triples to the correct positions in the output vector by exploiting the information created during the first pass. During the second pass, each subquery is also classified according to the intersecting/covering dichotomy.

The overall complexity of this phase is equal to $O(|P| + 2|Q| + |Q| + |\hat{Q}|) = O(|P| + 3|Q| + |\hat{Q}|)$: $|P|$ is due to the object locations indexing, $2|Q|$ is due to the two-pass approach, $|Q|$ is the cost to pay for the exclusive prefix sum performed between the first and the second pass needed to determine the subqueries locations in memory; finally, $|\hat{Q}|$ is related to the subqueries written out during the second pass.\\

\subparagraph*{Sorting (\SQ{}/\SQB{}).}

\begin{wrapfigure}{!hr}{0.4\columnwidth}
\vspace{-1.5em}
\begin{center}
    \includegraphics[width=0.4\columnwidth]{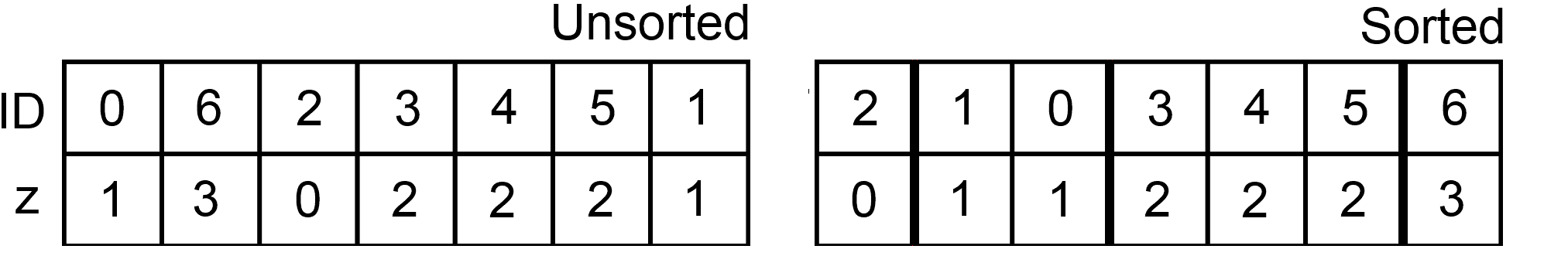}
    \caption{A simple example of a GPU-based sorting, based on the structure of vectors representation, of 8 entities according to their Morton codes. The discontinuities among the codes (thicker lines) determine the set of entities belonging to each cell.}
    \label{fig:sorting}
\end{center}
\vspace{-1em}
\end{wrapfigure}

Once object locations and subqueries are mapped to cells of $\cal C$, we sort them by the Morton codes of the cells, as illustrated in Figure \ref{fig:sorting}. The goal is to store tuples mapped to the same index cell in contiguous memory locations, thus enhancing the \emph{spatial locality} of each parallel block of threads working on subqueries and objects of a given active cell (i.e., a cell having at least one object) during the subsequent \emph{filtering} and \emph{decoding} phases (Sections \ref{sec:filterOverview} and \ref{sec:decodeOverview} respectively).

Indeed, when sorting the subqueries we distinguish between covering and intersecting ones, in order to support the optimizations discussed in Section \ref{sec:optimizations}. In practice, we handle the covering queries in a different way, since the GPU does not need to process them:  
after the sorting operation, all intersecting subqueries, which need to be processed, are placed at the beginning of the subqueries
structure of vectors.
%

Since the GPU sorting algorithm used throughout the pipeline will be the Radix Sort \cite{merrill11}, the complexity of the sorting step is $O(b \cdot (|P| + |\hat{Q}|)) \approx O(|P| + |\hat{Q}|)$, where $\hat{Q}$ denotes the subqueries set.

\subsection{Index Creation and Indexing in \QQ}
\label{par:indexDetOverviewQuad}

The key idea behind \QQ\ is to use a point-region (PR) quadtree as the backbone of its spatial index, exploiting the PR-quadtrees intrinsic ability to partition the space in differently sized parcels containing similar amounts of points. 

\subparagraph*{Index creation (\QQ{}).}

The goal of this phase is to create a space partitioning $\cal C$ over $\cal G$, according to Definition \ref{def:QuadPartitioning}, where each cell of $\cal C$ is a leaf PR-quadtree quadrant that does not contain more than $th_{quad}$ objects. 
This property gives an upper bound to the containment tests computed by each GPU thread in charge of processing a query over all the objects falling in an index cell.

\begin{algorithm2e}
\footnotesize
\DontPrintSemicolon

\Begin
{
$V_P \leftarrow GPUcalculateMortonHash(V_P,I_A,l_{max})$\; \nllabel{line:startSetup}\nllabel{line:calcZIndex}
$V_P \leftarrow GPUradixSort(V_P)$\; \nllabel{line:reorder}

\BlankLine \BlankLine

$I_A \leftarrow \{[0, |P|-1]\}$\; \nllabel{line:beginSetup}
${\cal C} \leftarrow \emptyset$\; 
$l \leftarrow 1$\; \nllabel{line:endSetup}

\Repeat{$(I_A \neq \emptyset) \land (l \leq l_{max})$}
{\nllabel{line:iterQuadBuildStart}
	$I \leftarrow GPUdetectQuadrants(V_P, I_A, l, l_{max})$\; \nllabel{line:detectQuadrants}
	$(I_A, {\cal C}) \leftarrow CPUcheckQuadrants(th_{quad}, I, l, l_{max}, {\cal C})$\; \nllabel{line:checkQuadrants}
	$l_{deep} \leftarrow l$\;
	$l \leftarrow l+1$\;
} \nllabel{line:iterQuadBuildEnd}

\BlankLine \BlankLine

$z_{map} \leftarrow GPUbuildZMap({\cal C},l_{deep})$\; \nllabel{line:buildZMap}
}

\caption{GPU-based PR-quadtree construction}
\label{lst:quadBuild}
\end{algorithm2e}
 
We observe that even if a space partitioning is determined according to local object densities for a particular tick, it can be often reused for consecutive ticks when the spatial distribution does not change significantly. 

Therefore we compute the spatial quadtree partitioning during the first tick, and repeat this partitioning if the objects spatial distribution change significantly, since this event might potentially hinder the performances by increasing the overall amount of containment tests to be computed per query.

The construction of the quadtree proceeds top-down in an iterative manner, starting from the $4$ equally sized quadrants that partition ${\cal G}$, and then splitting iteratively each quadrant containing more than $th_{quad}$ objects. The whole procedure is repeated level-wise, increasing the quadtree depth and splitting overpopulated quadrants if needed.
 
Algorithm \ref{lst:quadBuild} describes this iterative process. 
%
During the initial setup (lines \ref{line:startSetup} -- \ref{line:endSetup}), the function \textit{GPUcalculateMortonHash} (line \ref{line:calcZIndex}) computes the Morton codes \textit{z} of all the objects stored in the structure of vectors $V_p$ at the maximum quadtree level $l_{max}$. 
In practice, in this phase we consider a regular grid having $2^{l_{max}} \times 2^{l_{max}}$ cells. Morton codes
\textit{z} are computed in the same way as done in the \SQ{}/\SQB{} case, 
starting from the index $(i,j)$ of the regular grid where each object falls into.
Subsequently, $V_P$ is reordered by \textit{GPUradixSort} (line \ref{line:reorder}) according to the Morton codes \textit{z}. 
Note that, given the $z$-code at the maximum quadtree level $l_{max}$, we can 
determine the quadrant index $z'$ of any object at any level $l \leq l_{max}$ by simply
truncating the binary representation of the Morton code $z$ previously computed, which is
equivalent to calculating $z'=\frac{z}{4^{l_{max}-l}}$. 
It is worth considering that the object order 
obtained by this  sorting by $z$ is invariant for any level $l \leq l_{max}$ of the quadtree.
In other words, thanks to this sorting and the structural properties of quadtrees, objects contained in any quadtree leaf are memorized contiguously in $V_P$.

Subsequently, the algorithm initializes several variables: the set $\cal C$ of final leaves is initialized to $\emptyset$, the set $I_A$, containing the intervals of the quadrants to split, is initialized by inserting the interval related to the tree root, and, finally, the level $l$ from which the iterative construction starts is set to $l=1$ (lines \ref{line:beginSetup} -- \ref{line:endSetup}).

Then, the algorithm iteratively builds (line \ref{line:iterQuadBuildStart}) the quadtree level by level. \textit{GPUdetectQuadrants} (line \ref{line:detectQuadrants}) identifies the starting and ending positions (i.e., the intervals) of the $l$-level quadtree quadrants related to the $(l-1)$-level quadrants added to $I_A$ for splitting, and store such intervals in \emph{I}. Then, \textit{CPUcheckQuadrants} (line \ref{line:checkQuadrants}) determines which quadrants need further splitting at next level (their intervals are added to $I_A$) and which quadrants represent final leaves (their identifiers are added to $\cal C$). 
The process ends whenever no more quadrants need to be split (i.e., $I_A$ is empty) or the maximum possible quadtree level $l_{max}$ is reached (line \ref{line:iterQuadBuildEnd}). In the latter case, all the quadrants found at level $l_{max}$ are added to $\cal C$. We  postpone the description of \textit{GPUbuildZMap} (line \ref{line:buildZMap}) to a subsequent paragraph (see \textbf{\emph{Indexing moving objects and queries (\QQ)}}).

Functions \textit{GPUcalculateMortonHash}, \textit{GPUradixSort} and \textit{GPUdetectQuadrants} are entirely implemented on GPU. On the other hand, \textit{CPUcheckQuadrants} is executed on the CPU side, since the amount of quadtree quadrants created at each level are typically orders of magnitude lower than $|P|$.\\ 

\noindent \textit{{Simple running example}.}
Let us consider the example reported in Figure \ref{fig:quadConstruction}, where $l_{max} = 2$ and $th_{quad} = 1$.
During the \textit{Initialization} step, each object identified by an ID is associated with the Morton code \emph{z} of the cell $c \in {\cal C}_{l_{max}}$, where ${\cal C}_{l_{max}}$ denotes a uniform grid, associated with the deepest possible quadtree level $l_{max}$ (see the ``Initialization'' grid on the left of the figure).
The pairs $(ID, z)$  are stored in a table (see the ``Unsorted'' table). Subsequently, pairs are sorted 
according to the second elements, i.e., the Morton codes (see the ``Sorted'' table). Next, the algorithm proceeds building the quadtree, starting from \textit{Level 1}, creating iteratively new levels, until at least one quadrant requires to be split ($I_A \neq \emptyset$) or $l_{max}$ is reached.

At each level, the algorithm associates each object with a quadtree quadrant belonging to the currently considered level by computing the corresponding quadrant indices $z'$. In this regard see the ``Iterations'' table in Figure \ref{fig:quadConstruction}, where each row (after the second one) corresponds to an algorithm iteration working on a distinct level of the quadtree. On the right side of the same figure, we can also observe how the quadtree grows up at each iteration/level. 
More specifically, at each iteration the algorithm determines which quadrants need to be split. Objects falling in   quadrants to split are re-assigned by computing the new quadrant indices $z'$ (highlighted in red in the ``Iterations'' table). Quadrants that have not to be split are added to the set of $\cal C$ cells. Objects belonging to the latter kind of quadrants (highlighted in green in the ``Iterations'' table) can be ignored during the successive iterations (the ignored cells are highlighted in grey in the ``Iterations'' table), since these already belong to quadtree leaves. 

\begin{figure}[!h]
    \centering
	\includegraphics[width=0.95\columnwidth]{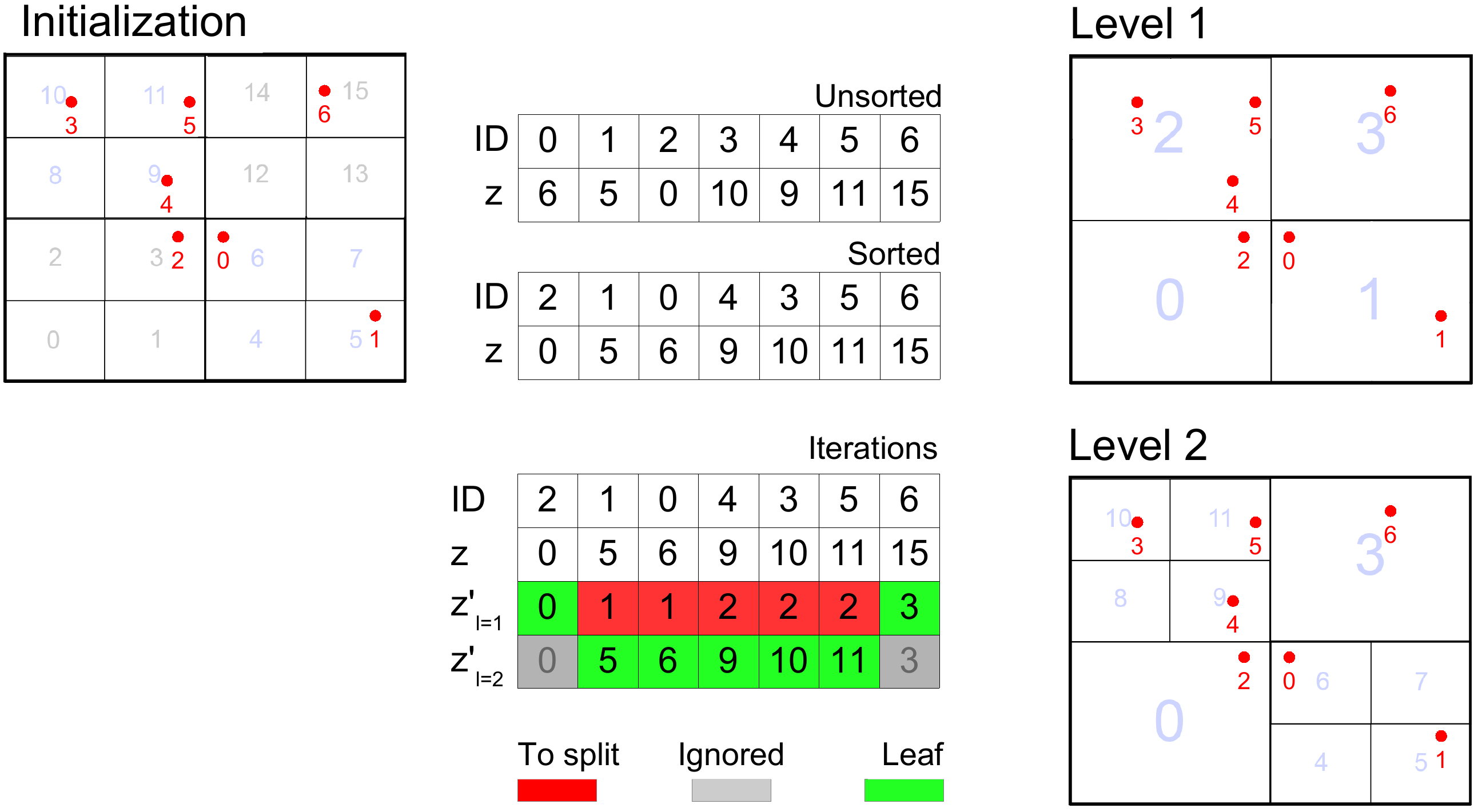}
    \caption{Example of quadtree construction with 7 objects, $th_{quad} = 1$ and $l_{max} = 2$.}
    \label{fig:quadConstruction}
\end{figure}

In the running example the algorithm stops at \textbf{Level 2}, since all the quadrants created at this level contain an amount of objects not greater than $th_{quad}$. Note that also the maximum quadtree level $l_{max} = 2$ 
is reached for 8 leaves of 10.\\

\noindent \textit{{Complexity}.}
The computation of a single Morton code has a fixed cost determined by the number of bits used for coordinate representation; therefore, \textit{GPUcalculateMortonHash} complexity is equivalent to $O(|P|)$. \textit{GPUradixSort} complexity is $O(b \cdot |P|) \simeq O(|P|)$, where $b$ represents the base value during sorting ($b \ll |P|$).
\textit{GPUdetectQuadrants} worst-case complexity is $O(l_{max} \cdot |P| + 2\sum_{l=0}^{l_{max}}4^l)$, where the first term is due to the amount of objects scanned in $V_P$ at each iteration, while $2 \cdot \sum_{l=0}^{l_{max}}4^l = 2 \cdot \frac{1-4^{l_{max}-1}}{1-4}$ represents the maximum amount of starting and ending indices - which has to be written out in memory - related to the $4^l$ quadtree quadrants at any level \textit{l}. We observe that the amount of quadrants created at each level is orders of magnitude lower than $|P|$, hence the related computational overhead is negligible. As a consequence, the average complexity can be safely approximated to $O(l_{max} \cdot |P| + 2\sum_{l=0}^{l_{max}}4^l) \simeq O(l_{max} \cdot |P|)$.
\textit{CPUcheckQuadrants} has a worst-case complexity equal to $\sum_{l=0}^{l_{max}}4^l = \frac{1-4^{l_{max}-1}}{1-4}$. Again, its complexity is practically negligible according to the above considerations.

Summing up, the complexity of the iterative process is dictated by the number of objects processed and the depth $l_{deep} \leq l_{max}$ reached in the quadtree construction, yielding $O(l_{deep} \cdot |P|)$. Since $l_{deep}$ is usually a low constant, the overall complexity can be approximated to $O(|P|)$.

\subparagraph*{Building a lookup table to map coordinates to cells (\QQ{}).}

%
The usual approach for finding the quadtree leaf that corresponds to the coordinates of an object 
would consist in traversing the tree from the root, recursively choosing the relevant node until a leaf is reached. Unfortunately this approach entails repeated irregular memory accesses and a non predictable number of operations for each leaf search. The second issue, in particular, would cause branch divergence and potential sub-optimal occupancy of GPU cores.

For this reason we use a different approach, characterized by a slightly larger memory footprint. Let us suppose that 
the deepest level created in a quadtree ${\cal C}$ is $l_{deep}$, $l_{deep} \leq l_{max}$. Thus we virtually 
divide the space covered by ${\cal C}$ according to a uniform squared grid composed of $2^{\cdot l_{deep}} \times 
2^{\cdot l_{deep}}$ cells, and denote it by ${\cal C}^{l_{deep}}$. In other words, we cover ${\cal C}$ such that each 
quadtree leaf created at level $l_{deep}$ corresponds exactly to a single cell in ${\cal C}^{l_{deep}}$. Thanks to the PR-quadtree properties, any quadtree leaf at a level $l$, $l \leq l_{deep}$, corresponds to the union of $4^{(l_{deep} - l)}$ contiguous cells of ${\cal C}^{l_{deep}}$. Therefore, a mapping between ${\cal C}^{l_{deep}}$ cells and ${\cal C}$ cells can be easily established by means of a lookup table $z_{map}$, which maps each ${\cal C}^{l_{deep}}$ cell, identified by a pair $(i,j)$, to the ${\cal C}$ cell containing it. 

The idea behind this approach is exemplified in Figure \ref{fig:zmap}. The example is derived from the one in Figure \ref{fig:quadConstruction}, and therefore $l_{deep} = 2$. Each $\cal C$ cell (quadtree leaf) is identified by a pair $(l,z)$ (an integer is indeed sufficient to store each pair), where \textit{l} is the leaf level and \textit{z} its Morton code at level \textit{l}, whereas each cell in ${\cal C}^{l_{deep}}$  is associated with the pair $(l,z)$ identifying  the cell of $\cal C$ containing it. We can observe that 4 distinct ${\cal C}^{l_{deep}}$ cells are mapped to the same  $\cal C$ cell $(1,0)$, and other 4 distinct ${\cal C}^{l_{deep}}$ cells are mapped to the same  $\cal C$ cell $(1,3)$.

Therefore, given any pair of coordinates, it is possible to find the associated $\cal C$ cell by first computing the associated ${\cal C}^{l_{deep}}$ cell index, namely a pair $(i,j)$, and then performing a lookup in $z_{map}$. Both operations have constant complexity, even though we have to mention that the performance related to the lookups in $z_{map}$ heavily depends on the ability to exploit the GPUs caching capabilities. Indeed, $z_{map}$ may have a relevant size - depending on $l_{deep}$. In light of this, it is important the \emph{memory layout} of $z_{map}$ to enhance data locality. 

%

As regards the \emph{memory layout} of the bidimensional array $z_{map}$, instead of 
using the typical row-major order memory layout, we access it
according to the Morton code obtained from index pairs $(i,j)$ used to access the array.
Since all objects and queries are first associated with
the Morton code of the cell $(i,j)$ in ${\cal C}^{l_{deep}}$ which contains them,
and then are \emph{sorted} by this code, during the indexing operation described below we access 
$z_{map}$ by exploiting temporal and spatial locality.
This is because when we scan objects and queries that are memorized nearby, we also access nearby 
elements in $z_{map}$.

The initialization of $z_{map}$ is performed entirely on GPU (function $GPUbuildZMap$, line \ref{line:buildZMap} in Algorithm \ref{lst:quadBuild}), by assigning each ${\cal C}$ cell (quadtree leaf) to a GPU streaming multiprocessor, which in turn initializes the interval of cells (elements of the lookup table) in ${\cal C}^{l_{deep}}$ contained by the ${\cal C}$ cell assigned.

The complexity of $GPUbuildZMap$ is $O(|{\cal C}| + |{\cal C}^{l_{deep}}|)$ and, in practical terms, negligible. 

\begin{figure}[!ht]
	\centering
    \begin{minipage}{0.45\columnwidth}
    \centering
	\includegraphics[width=0.7\columnwidth]{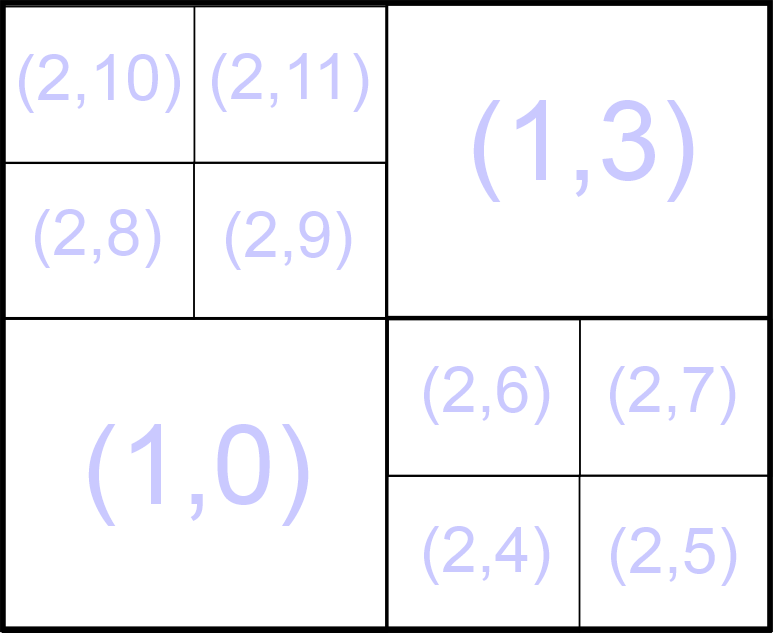}
	\end{minipage}
    \begin{minipage}{0.45\columnwidth}
    \centering
    \includegraphics[width=0.7\columnwidth]{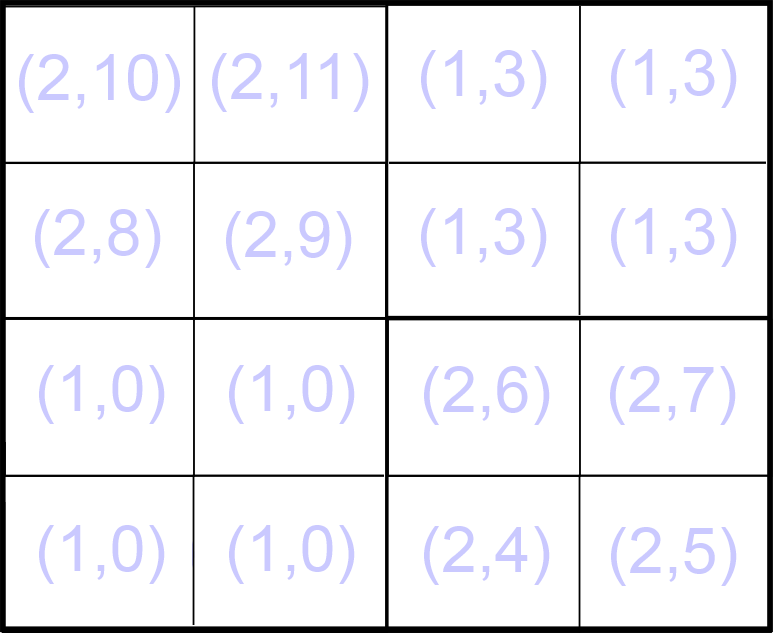}
\end{minipage}
	
\caption{Example of the mapping established by $z_{map}$ between the quadtree-induced grid $\cal C$ (left side) and the uniform grid ${\cal C}^{l_{deep}}$ (right side) related to the quadtree deepest level.}
\label{fig:zmap}
\end{figure}


\subparagraph*{Indexing moving objects and queries (\QQ{}).}
%

The goal of this phase is to map objects locations and queries to $C$ cells (quadtree leaves). Each object location is mapped to a single cell while each query can be potentially mapped to multiple cells (Definition \ref{def:objPosMap}). 

To convert the position of all objects in $P$ to cells identifier $c$ of ${\cal C}$, their 2-dimensional coordinates are first mapped to grid coordinates $(i,j)$ in the ${\cal C}^{l_{deep}}$ grid, where $l_{deep}$ is the deepest level of ${\cal C}$. Subsequently, the Morton codes identifying the cells of ${\cal C}^{l_{deep}}$ are derived from $(i,j)$. Then, objects are \emph{sorted} according to such Morton codes in order to exploit caching when subsequently accessing $z_{map}$, where $z_{map}$ is used to retrieve the final quadtree cell identifier $c=z_{map}[i,j]$, $c \in {\cal C}$, in which the objects fall. 
We remark that objects remain sorted after such mapping thanks to quadtrees structural properties: this will be exploited during the filtering and decoding phases (see Sections \ref{sec:filterOverview} and \ref{sec:decodeOverview}), since query processing happens at cell level.

To convert a range query, characterized by a rectangular region, we need to identify all the relevant cells in ${\cal C}$, i.e., all the cells that spatially intersect the query. 
This process entails to identify, for each query $q$, a set of \textit{subqueries}, each corresponding to the spatial restriction of the rectangular region of $q$ to a relevant cell in ${\cal C}$.
More formally, we have $g(q)=\{c_1,c_2,\ldots,c_n\} \subseteq {\cal C}$ (see Definition \ref{def:objPosMap}), where $q$ intersects or covers each $c_i \in {\cal C}$. We thus refer to each pair $(q,c_i)$ as a \textit{subquery} of $q$.

First, as in the objects case, queries are first associated with a ${\cal C}^{l_{deep}}$ cell, namely their Morton codes, through their reference corner (in case part of their spatial extent falls outside the MBR $\cal G$, only the area in common with $\cal G$ is considered), and then sorted accordingly to such codes in order to exploit caching when subsequently accessing $z_{map}$.

Then, to obtain the subqueries, we start by identifying all the 
${\cal C}^{l_{deep}}$ cells intersected by the query. We map each of these cells identified by a pair $(i,j)$ 
to the corresponding $c \in {\cal C}$ cell by exploiting $z_{map}$. Depending on the spatial distribution, it is very likely to have multiple 
cells of ${\cal C}^{l_{deep}}$ that intersect the range query, and are thus mapped to the same ${\cal C}$ cell. 
This behavior could create \emph{duplicate} subqueries, i.e., the same query mapped multiple times to the same cell of $C$. 
Figure \ref{fig:queries duplicates handling} illustrates the problem and sketches our solution to avoid the presence of multiple subqueries mapped to the same $\mathcal{C}$ cell. 
In the left picture of the figure, we can see how the original query $q_1$ falls over multiple $\cal C$ cells (specifically, 6 distinct cells). Among these, we consider the intersection between $q_1$ and the $\cal C$ cell $(1,3)$ (yellow area). In the right picture, which illustrates the uniform grid ${\cal C}^{l_{deep}}$ associated with ${\cal C}$ through $z_{map}$,  we can note that there are multiple cells on the right-upper part of $q_1$ that map to cell $(1,3)$ in $\cal C$ (specifically, 4 distinct cells of ${\cal C}^{l_{deep}}$). 
Therefore, $q_1$ would yield 4 subqueries 
that map to the same $\cal C$ cell $(1,3)$. To avoid duplicates, we always select the subquery 
having the minimal grid coordinates (highlighted in green).

\begin{figure}[!ht]
	\centering
    \begin{minipage}{0.49\columnwidth}
    \centering
	\includegraphics[width=0.7\columnwidth]{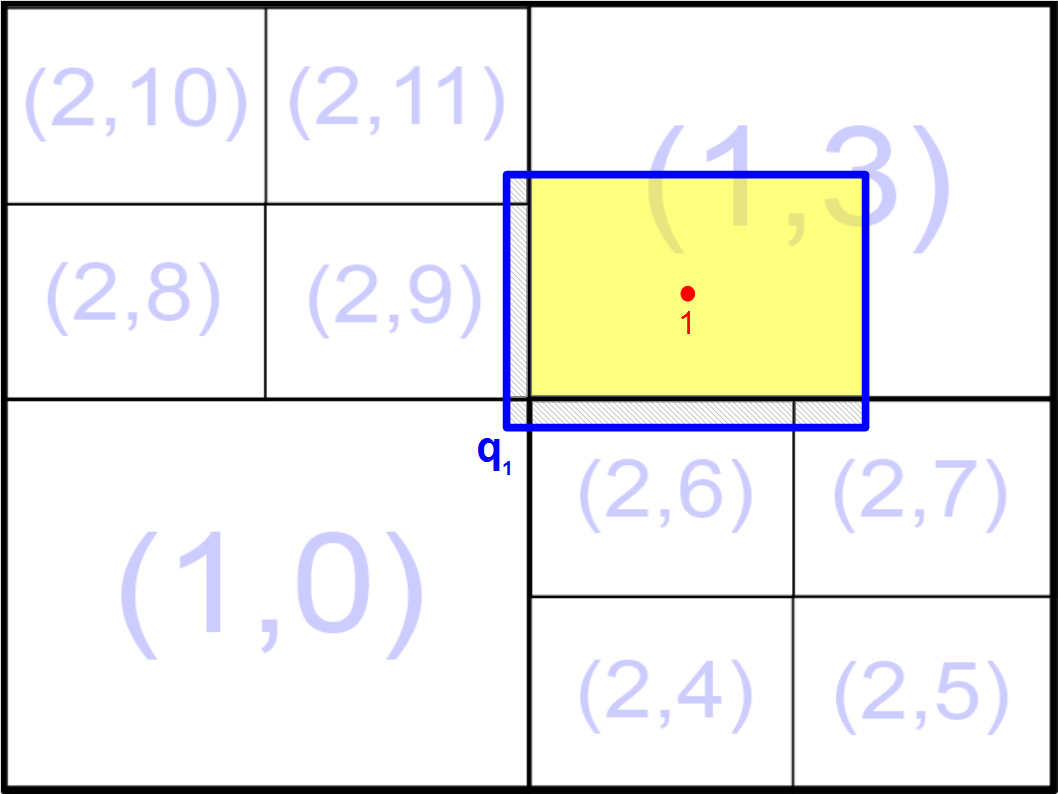}
	\end{minipage}
    \begin{minipage}{0.49\columnwidth}
    \centering
    \includegraphics[width=0.7\columnwidth]{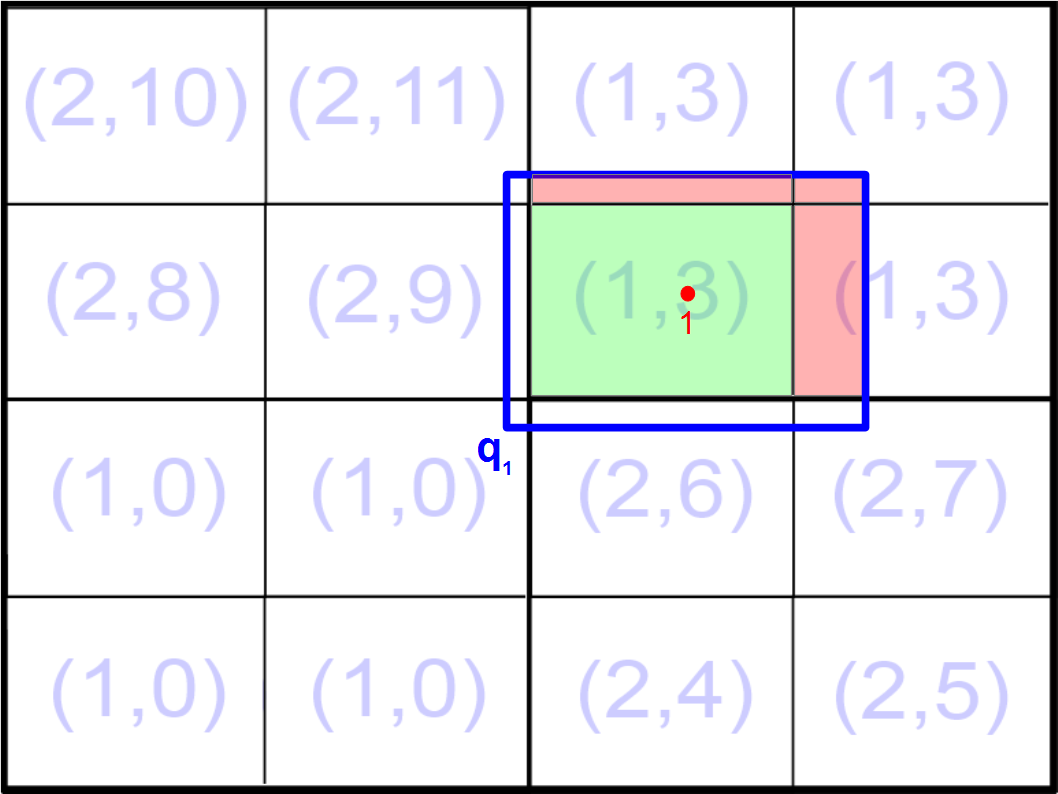}
\end{minipage}
	
\caption{Query indexing example with \QQ.}
\label{fig:queries duplicates handling}
\end{figure}


%

%

The queries are indexed, similarly to \SQ\ (\SQB), in two separate phases. During the first phase the amount of subqueries per each original query is determined. In order to determine the memory location where each subquery will be written, an exclusive prefix sum is performed over the vector containing the amounts of subqueries per query. Then, in the second phase, subqueries are actually written using the information computed during the first phase, and classified according to the \textit{intersecting/covering} dichotomy.

The overall complexity of the indexing phase can be expressed in the following terms:

\begin{itemize}
\item for what is related to the sorting operations needed to optimize the accesses in $z_{map}$, we have $O(|P| + |Q| + b \cdot (|P| + |Q|)) \approx O(2(|P| + |Q|))$: $|P|$ and $|Q|$ are due to $l_{deep}$ Morton codes computations while $b \cdot (|P| + |Q|)$ relates to the actual sorting performed, by means of Radix Sort, over \emph{P} and \emph{Q}.

\item for what is related to the subsequent operations, we have $O(|P| + O(2|Q| \cdot |{\cal C}^{l_{deep}}| + |Q| + 2|\hat{Q}|)$: $|P|$ relates to the lookups in $z_{map}$, $2|Q| \cdot |{\cal C}^{l_{deep}}|$ relates to the query indexing which happens in two separate phases ($|{\cal C}^{l_{deep}}|$ is due to the amount of ${\cal C}^{l_{deep}}$ cells spanned by an original query in the \textit{worst} case) and $2|\hat{Q}|$ relates to the subqueries written during the second phase (lookups in $z_{map}$ are included in the complexity), noting that $|\hat{Q}| = |Q| \cdot |{\cal C}^{l_{deep}}|$ in the worst case.
\end{itemize}

In light of these considerations, the amount of subqueries to be checked during indexing may be relevant, therefore we remark the importance of exploiting caching when accessing $z_{map}$.

\subparagraph*{Sorting (\QQ{}).}

Once subqueries are mapped to $\cal C$ cells, we sort the associated augmented tuples to store  them in contiguous memory locations. The reason of this phase is analogous to the sorting carried out in the \SQ\ and \SQB\ cases.


Note that, unlike the \SQ\ and \SQB\ cases, for \QQ\ we do not need to re-sort the moving objects.
The sorting done during the indexing phase, according to the cell identifiers of ${\cal C}^{l_{deep}}$ grid, 
is enough to guarantee locality during the following query processing phase, thanks to
the quadtrees structural properties. 

As regards subqueries, we have to sort them since there is no guarantee about the order in which they are written in global memory during the indexing phase.
Consequently, the structure of vectors associated with the subqueries tuples, $\hat{Q}$, is sorted on GPU by means of Radix sort according to the identifier associated with the cell
\footnote{While in the \SQ\ and \SQB\ cases this identifier is represented by an integer storing grid coordinates,  in the \QQ\ case it is a pair $(l,z)$, which is indeed conveniently stored as an integer.}.
Moreover, in order to support the optimizations discussed in Section \ref{sec:optimizations}, each identifier is augmented so to signal whether a subquery is either covering or intersecting. In this way, after the sorting operation, all intersecting subqueries are placed at the beginning of their structure of vectors.
%

Since the sorting algorithm is Radix Sort, the complexity of the sorting step is 
$O(b \cdot (|\hat{Q}|)) \approx O(|\hat{Q}|)$.

\subsection{Filtering}
\label{sec:filterOverview}
The goal of the filtering phase is to compute range queries over object locations, and store the
containment test outcomes (i.e., which object locations are contained in each query range) conveniently. Since, by definition,
covering subqueries entirely cover the cell onto which they fall, the filtering phase can be actually limited
to intersecting subqueries, delegating the processing of the former type to the optimization described
in Section \ref{sec:optimizations}.

In this context we conveniently denote by ${\cal C}_\alpha \subseteq {\cal C}$ the set of active cells, i.e., those cells containing at least one object.
 
Both \QQ{} and \SQ{} store the containment test outcomes in form of bitmaps (one per \emph{active} cell), which will be decoded at a later stage in order to obtain a final compact representation of the positive containment test outcomes.

Filtering is performed in parallel: each active cell in $\cal{C}_{\alpha}$ is assigned to a block of GPU threads to obtain a bitmap which refers to object locations and subqueries falling in the corresponding cell.

In the last part of this subsection we also detail the simplifications adopted by the \SQB{} filtering algorithm. In the experimental section we will use \SQB\ as a baseline to assess the benefits of using the bitmaps and an additional (decoding) phase needed to extract the final query results from these.

\subparagraph*{Bitmap layouts.}
Bitmaps are arranged in memory by using two different layouts across the following phases, since each layout better fit specific kinds of operations on GPU.
For each active cell $c \in {\cal C}_{\alpha}$, the filtering phase  initially compute a 2D bitmap $B^{c}$ characterized by an interlaced column-wise layout (Figure \ref{fig:bitmapInterlaced}.a), where each column $q_i$ refers to a single query and each row $b_j$ refers to a fixed block of $w$ object locations inside the cell. The width of each column is $w$ bits (here we assume $w=32$) and the content of a $w$-bit word corresponding to $q_i$ and $b_j$ indicates if the object locations associated with block $b_j$ are contained in the extent of query $q_i$. Thus, a single bitmap element $B_{n,m}$, where $n$ and $m$ are the row and column at bit level, represents the containment test outcome between the ($m / w$)-th query and the ($(n \cdot w) + (m \text{ } mod \text{ } w)$)-th object location.

\begin{figure}[h]%
\centering
\subfloat[Bitmap layout with \textit{interlaced} bit-vector words.]{\includegraphics[width=0.45\columnwidth]{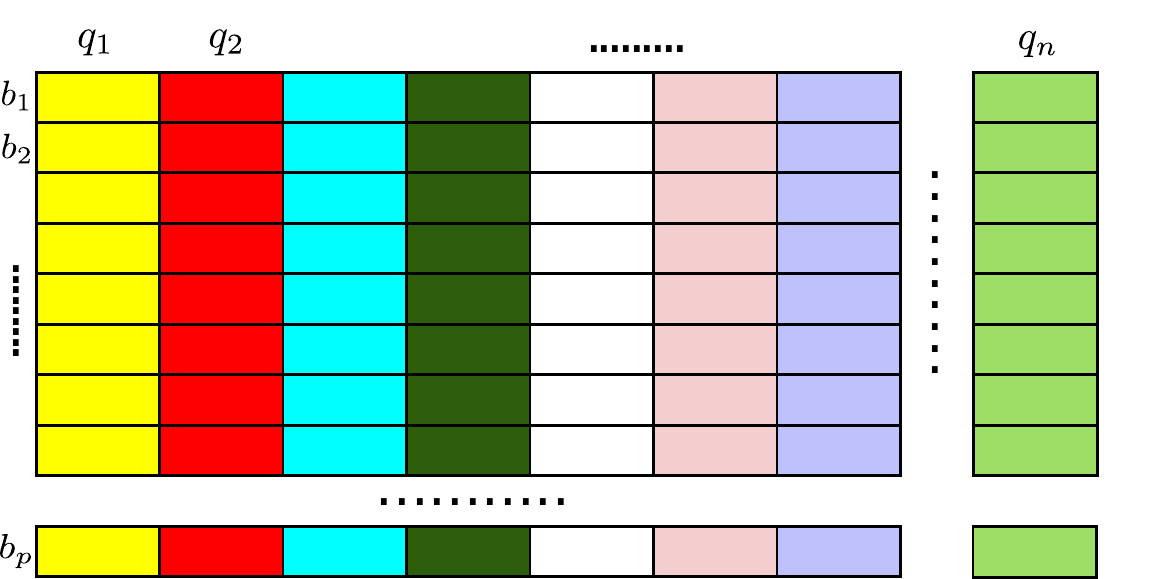}}
\hspace{0.5em}
\subfloat[Bitmap layout with \textit{contiguous} bit-vector words.]{\includegraphics[width=0.45\columnwidth]{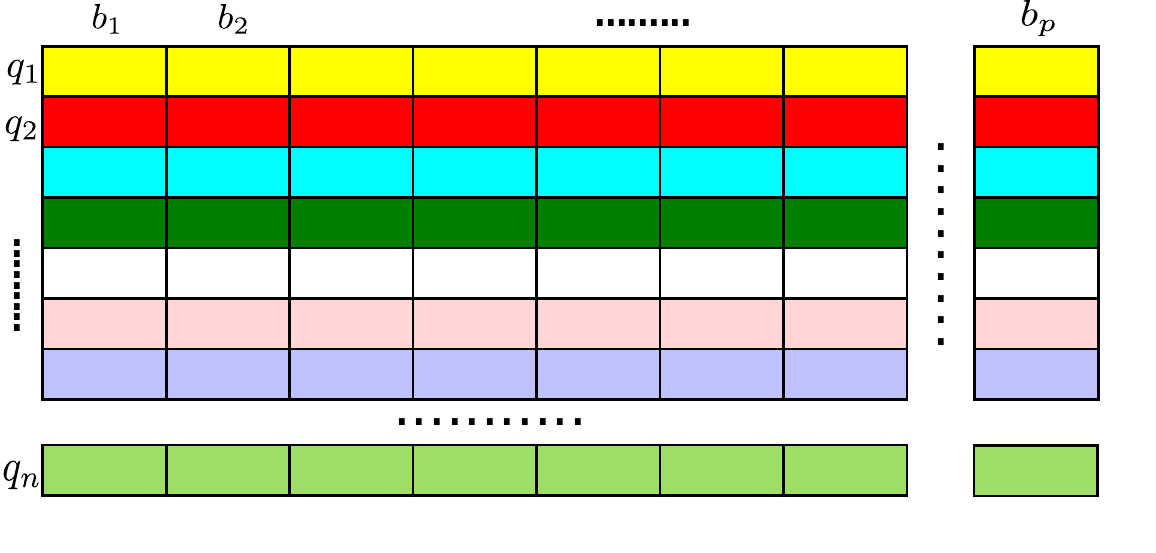}}
\caption{Bitmap layouts used during the processing.}
\label{fig:bitmapInterlaced}
\end{figure}

This interlaced layout choice facilitates coalesced memory accesses when threads in the same warp are in charge of computing containment tests of a consecutive group of subqueries (see \emph{Interlaced bitmap generation}). Indeed, these threads can write results to consecutive memory positions, taking advantage of coalesced memory accesses. Since threads in the same warp perform exactly the same operations, and each thread is in charge of writing a distinct 32-bit word, this solution eliminates the need of any synchronization mechanisms at thread block or global levels. 


While the interlaced layout entails benefits when the bitmap is produced, it hinders the extraction of all the containment test outcomes referred to the same subquery. For this reason, after production, each interlaced bitmap is transposed word-wise to improve the memory throughput when the bitmaps are decoded to  extract the final results. We will refer to this transformation as the \textit{linearization} of interlaced bitmaps. Indeed, in the row-wise layout resulting from the transformation (Figure \ref{fig:bitmapInterlaced}.b), single containment test outcomes are linearly indexed and bit-vectors associated with each subquery have their words arranged consecutively in memory, which favours subquery-wise read coalescing during the decoding phase. The linearization transformation can be expressed as a massively-parallel operation which is efficiently performed on GPU.

\subparagraph*{Filtering -- Interlaced bitmap generation.}
During this stage we divide query result computation to exploit three different kind of parallelism allowed by GPUs. 

\emph{Block parallelism} allows to process independent tasks. Since we are considering subqueries, which are restricted to a specific index cell by definition, the computation of the results in different cells can proceed independently, producing distinct result bitmaps. Thus, active (non-empty) index cells $c \in \cal C_{\alpha}$ are assigned to distinct \emph{blocks} of GPU threads.

Each block of GPU threads is executed asynchronously by the same \emph{streaming multiprocessor} (SM). \emph{Thread parallelism} allows for cooperation among threads in the same block.  Each thread in a block is in charge of computing a distinct subquery that is present in the index cell assigned to the block. Since each bitmap is common for all the subqueries(threads) in a cell(block), the cooperation among threads is used to ensure coordination when writing out the containment test outcomes (0/1) in the bitmap.

Whenever possible it is strongly suggested to orchestrate the thread scheduling to hide memory access latency by having an amount of threads per thread block exceeding the amount of cores per single streaming multiprocessor. However, only subsets of threads can run in parallel at a given time. These subsets of $sz_{warp}$\footnote{32 threads per warp in current GPUs.} synchronous and data parallel threads are called \emph{warps}. Thanks to synchronous execution, \emph{warp parallelism} allows to avoid synchronization operations. Furthermore, threads in the same warp benefit from coalesced memory accesses when they access consecutive (or identical) memory positions, so that several memory accesses are combined in a single transaction.

In our solution all the threads of a warp access the device memory in an optimal way: they read the same input data (object locations) synchronously (this exploits GPU caching), access them consecutively (subqueries, spatial locality, this exploits coalescing) and, thanks to the interlaced bitmap layout, write simultaneously results ($w$-bits bitmap words) to consecutive memory locations (this entails coalesced memory access).

To better explain the latter point, we illustrate in Figure \ref{fig:create_interbitmap} the role of different threads in a thread block during the creation of an interlaced bitmap: each group of
$sz_{warp}$ columns is collectively updated by a warp of threads in the thread block.
\begin{figure}[ht]
  \centering
  \includegraphics[width=0.5\columnwidth]{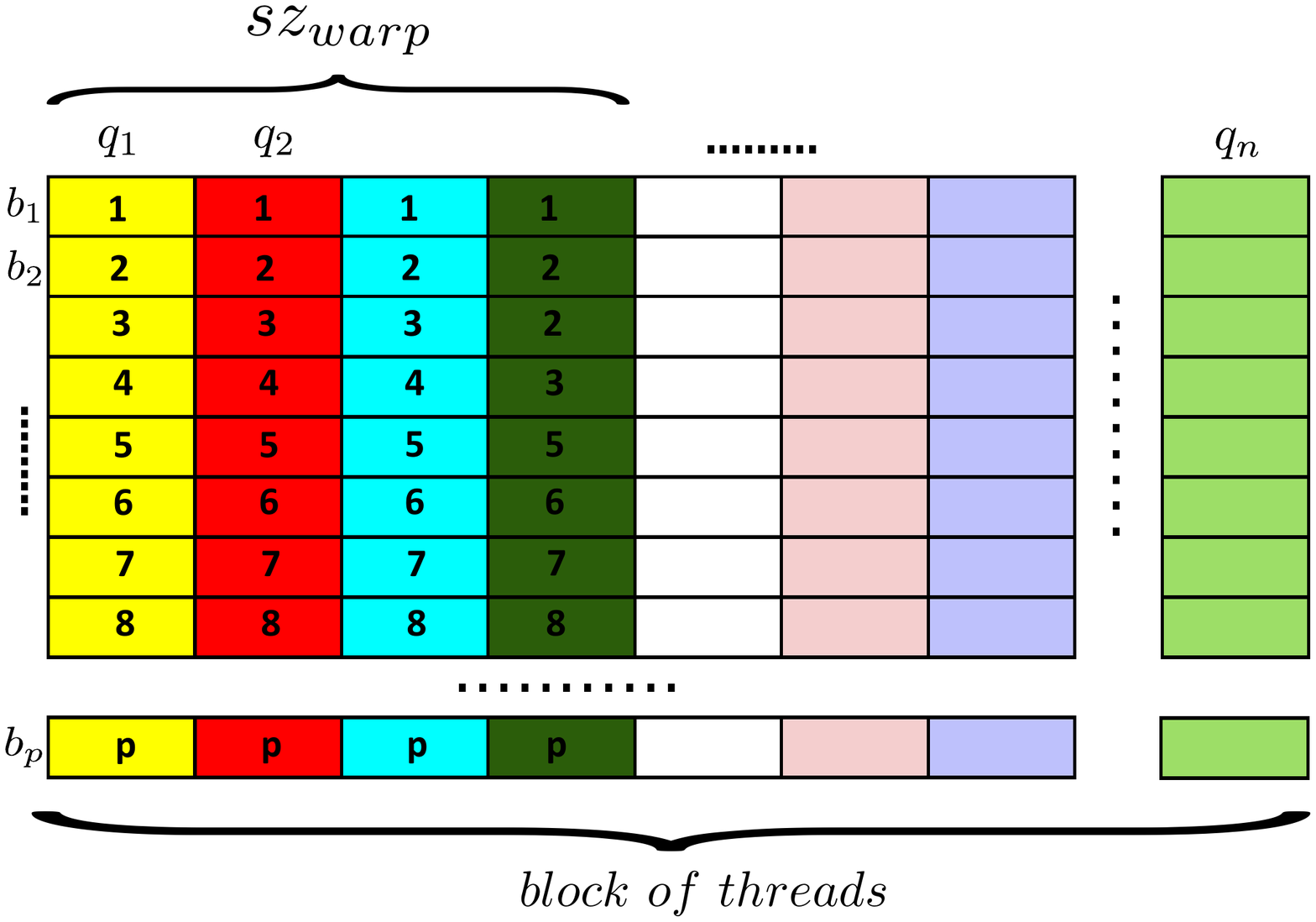}
  \caption{Collective creation of a interlaced bitmap by a block of GPU threads.}
  \label{fig:create_interbitmap}
\end{figure}
Each column contains a set of 32 bit-wide words, $b_1, \ldots, b_p$,
associated with a subquery $q$. The first words (the $b_1$'s)
associated with the various subqueries, and computed by the warp threads,
are stored simultaneously in memory. 
The same holds for the second block of words (the $b_2$'s), which is stored immediately after, and so on.  
The bitmap words updated simultaneously by the threads are stored consecutively thanks to the
interlaced layout of the bitmap. This permits the writes to be \emph{coalesced} .

\begin{algorithm2e}[h]
\DontPrintSemicolon
\SetInd{0.7em}{0.7em}
\Begin
{
$numPoints \leftarrow 0$\;
$wordIndex \leftarrow 0$\;
$wordBitmap \leftarrow 0$\;

\BlankLine \BlankLine

\ForEach{$c \in \cal{C}_\alpha$ $\lstparallel{block}$}
{\nllabel{cellPar}
	\ForEach{$q \in c$ $\lstparallel{thread}$}
	{\nllabel{queryPar}
		\ForEach{$p \in c$}
		{\nllabel{readPoints}
		
			$numPoints \leftarrow numPoints + 1$\;
			\If{$p \in q$}
			{
				$setBit(wordBitmap, p)$\; \nllabel{wordWrites}
			}	
			\If{$numPoints \text{ mod } 32 = 0$}
			{
				$writeBitmap(wordBitmap, wordIndex, q)$\; \nllabel{bitmapWrites}
				$wordBitmap \leftarrow 0$\;
				$wordIndex \leftarrow wordIndex + 1$\;
			}
		}
	}
}
}

\caption{\QQ{} and \SQ{} filtering phase.}
\label{lst:filtering}
\end{algorithm2e}

The pseudocode in Algorithm \ref{lst:filtering} illustrates the main points of the interlaced bitmap generation. Distinct \emph{blocks} of GPU threads process in parallel active index
cells $c \in {\cal C}_{\alpha}$ (line \ref{cellPar}). Each thread in a block is in charge of computing the results of a distinct subquery present in $c$ (line \ref{queryPar}).

%

All threads read the same sequence of object
positions and update a private \emph{32 bit-wide} register that
contains the bitwise information about the presence/absence of 32 distinct object locations in its own range query (line \ref{wordWrites}).
When the threads in a warp have completed the update of the current word (i.e., they have finished to compute a block of 32 containment tests or they have computed all the blocks), all threads proceed by flushing the content of their private registers simultaneously to the global device memory at the right memory displacement (line \ref{bitmapWrites}). The computation goes on until all the subqueries have been computed.

The execution of the inner loop (line \ref{readPoints}) is scheduled by the GPU at warp level: it depends on resource availability and memory access latency, but threads in the same warp are granted to be synchronous. For example, $wordBitmap$ will be completed simultaneously for all the threads in the same warp.

We finally note that all the threads of any warp access the device memory in an optimized way: they read the same input data synchronously (object positions, line \ref{readPoints}) or consecutively (subqueries, line \ref{queryPar}), thus always exploiting data spatial locality. 
Moreover, all threads write simultaneously words that are stored consecutively in memory, thus coalescing the writes and boosting the overall GPU global memory throughput (line \ref{bitmapWrites}).

\subparagraph*{Filtering -- Bitmap linearization.}

The goal of this operation is to transform each bitmap from the interlaced column-wise layout to 
the linearized row-wise one, so that bitmaps can be more efficiently processed during the subsequent decoding phase (Section \ref{sec:decodeOverview}). This transformation is performed on GPU and is depicted in
Figure \ref{fig:linearization}. The Figure shows the work of a single warp composed of $sz_{warp}$ threads.
\begin{figure}[ht]
  \centering
  \includegraphics[width=.6\columnwidth]{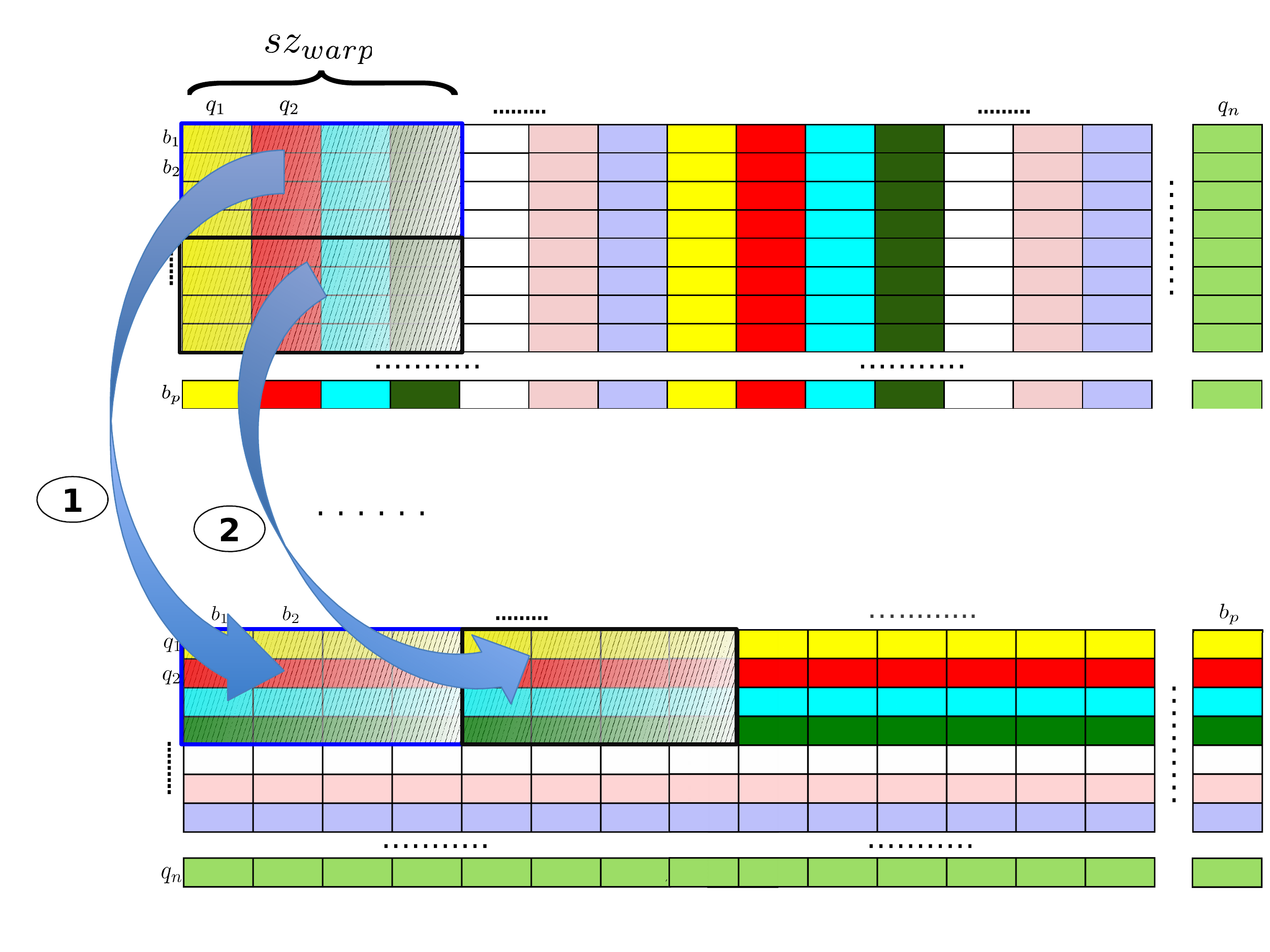}
  \caption{Collective linearization of a group of columns by a warp of GPU threads.}
  \label{fig:linearization}
\end{figure}
All the synchronous threads cooperate to transpose blocks of words,
one block at a time. Each block is made up of $sz_{warp} \times
sz_{warp}$ words. The linearization of a block of words is carried out
in two steps:
\begin{enumerate}
\item First, all the $sz_{warp}$ threads read in parallel the bitmap
  words associated with the subqueries assigned to them (a row of the
  input block of words at a time). The threads proceed until the
  end of the block (for $sz_{warp}$ rows of the block). While reading
  the words associated with subqueries, each thread incrementally prepares a
  linearised block of words in a temporary buffer stored in the fast
  shared memory available to each SM.

\item Once the input block has been read completely and linearized in
  the shared memory, the warp threads start their second step, where
  they move the linearized block to the device global memory. The only
  difference is that they collaborate by writing portions of the
  bitmaps associated with each subquery in parallel. First they write the
  first row of the linearised block in parallel, then the second, and so on.
  This, in turn, entails coalesced writes.
\end{enumerate}


Even if the shared memory has limited size, the block-wise linearization does not saturate the shared memory thanks to the small size of the data that are linearized simultaneously.


%
%

\subparagraph*{Filtering phase in \SQB.}
The \SQB\ algorithm does not use bitmaps and computes the query results on the fly. Therefore, its filtering phase is simpler and does not require subsequent linearization and decoding phases. The pseudocode in Algorithm \ref{lst:filteringSQB} illustrates this simpler strategy.

\begin{algorithm2e}[h]
\DontPrintSemicolon
\SetInd{0.7em}{0.7em}

\Begin
{
\ForEach{$c \in \cal{C}_\alpha$ $\lstparallel{block}$}
{\nllabel{cycleBlockSQB}
	shared $resultBuffer \leftarrow \emptyset$\;
	\ForEach{$q \in c$ $\lstparallel{thread}$}
	{
		\ForEach{$p \in c$}
		{
			\If{$p \in q$}
			{ \nllabel{computeSQB}
				$appendResult(\textit{resultBuffer}, q, p)$\; \nllabel{appendSQB}
			}
			\If{$full(\textit{resultBuffer}) = true$}
			{ \nllabel{fullSQB}
				$flush(\textit{resultBuffer})$\; 
			}
		}
	}
}
}

\caption{\SQB{} filtering phase}
\label{lst:filteringSQB}
\end{algorithm2e}

Each active cell is assigned to a single thread block (line \ref{cycleBlockSQB}). Query results (pairs of object locations and subqueries such that the object location is contained in the subquery range) are computed in parallel (line \ref{computeSQB}) and immediately appended to a buffer in shared memory (line \ref{appendSQB}) common to all the threads in the block. The threads access and update the buffer by means of \emph{atomic operations} in order to guarantee its consistency. Once the buffer contains an amount of results greater than a fixed threshold (line \ref{fullSQB}), the threads synchronize and cooperatively flush its content out to global memory. In order to guarantee the consistency of the data written to the global memory a result counter, shared among all the thread blocks, is accessed and updated atomically.

\subparagraph*{Filtering -- Complexity.}

The \emph{interlaced bitmap generation} (Algorithm \ref{lst:filtering}) represents the time dominant part of the filtering phase. The filtering complexity, determined by the amount of containment tests to be computed, is:
\begin{equation}
\label{eq:eqFiltComp2}
O(\sum_{c \in \cal C_{\alpha}} |\hat{Q}_I^c| \cdot |P^c|),
\end{equation}
where  $C_{\alpha} \subseteq {\cal C}$ is the set of  active grid cells,  $P^c$ the set of object locations in $c$ and $\hat{Q}_I^c$ the intersecting subqueries associated with $c$. 

In general we observe that decreasing the grid cells size yields a smaller number of containment tests, even if the number of intersecting subqueries to manage is larger. 
An arbitrary decrease, however, has negative side effects, such as the fragmentation of the intermediate results in a large number of small bitmaps, which in turn influences negatively the overall running time due to inefficient computational resource usage and scattered memory accesses. Moreover, an arbitrary cell size decrease may induce an intersecting subqueries increase rate eclipsing the decrease rate related to the average amount of objects per active cell, therefore raising the complexity at some point. In light of these considerations we argue there is a trade-off between decreasing the overall number of operations executed and optimizing parallelism and memory access costs.

The \emph{bitmap linearization} has linear complexity  with respect to the number of bitmaps words. Since each bit corresponds to an intersection test, the two subphases has the same complexity.

\subsection{Bitmap decoding}
\label{sec:decodeOverview}

The end product of the filtering phase of both \SQ{} and \QQ{} consists of a set of linearized bitmaps, one per active cell, containing both the positive and the negative containment test outcomes, related to object locations and subqueries associated with the active cells. The goal of the decoding phase is to process such bitmaps in order to extract the final query result set, i.e., the positive occurrences in the bitmaps. 

Accessing the bitmaps content in order to extract the positive occurrences represents a memory and computationally intensive task which, thanks to the linearized layout, can be efficiently and conveniently parallelized on GPU.

The decoding phase, common to \SQ{} and \QQ{}, proceeds as follows: for each intersecting subquery a list of objects identifiers is generated (according to the problem definition in Section \ref{sec:problem}), where each identifier represents a positive occurrence. Since each active cell (bitmap) represents a single task, the decoding operation can progress by decoding the tasks in chunks: this allows us to transmit to the CPU the information related to a previously decoded chunk of bitmaps while the GPU progresses by decoding the next chunk of unprocessed bitmaps. This also allows us 
to overlap computations carried on the GPU with I/O transfers from GPU to CPU. We highlight that the result lists related to intersecting subqueries originating from the same query have no results in common and preserve the identifier of the original query, so it is trivial to merge them to obtain the final result set.

The pseudocode in Algorithm \ref{lst:decoding} illustrates the GPU part of this strategy.

\begin{algorithm2e}[h]
\DontPrintSemicolon
\SetInd{0.7em}{0.7em}
\Begin
{
\ForEach{$bitmap$ $\lstparallel{block}$}
{\nllabel{cycleBlock}
	\ForEach{$q \in Q_{bitmap}$ $\lstparallel{warp}$}
	{\nllabel{cycleWarp}
		$q_{id} \leftarrow loadQueryID(q)$\; \nllabel{loadQueryID}
		\textbf{shared} $q_{bv} \leftarrow loadQueryBitVector(q,bitmap)$\; \nllabel{loadBV}
		$resultSet_q \leftarrow linearScan(q_{bv})$\; \nllabel{resSetScan}
		
		\BlankLine \BlankLine
		
		\ForEach{$r \in resultSet_q$ $\lstparallel{thread}$}
		{
			$p \leftarrow loadPoint(r,bitmap)$\; \nllabel{loadPointID}
			$writePID(p, bitmap)$\; \nllabel{resWrite}
		}
		
		\BlankLine \BlankLine
		
		$writeQueryDetails(q_{id},|resultSet_q|, bitmap)$\;
	} \nllabel{endCycleWarp}
}
}

\caption{Decoding phase}
\label{lst:decoding}
\end{algorithm2e}

Each bitmap refers to a specific index active cell and is decoded by a specific thread block (line \ref{cycleBlock}). Each intersecting subquery referred by the bitmap is assigned to a \textit{warp} (line \ref{cycleWarp}), so that each warp is in charge of several subqueries. Since threads in the same warp are synchronous, the use of warp level granularity allows to safely avoid the use of synchronization mechanisms inside the loop.

Threads in the same warp transfer (line \ref{loadBV}) the words composing the bit vector of the currently considered subquery from device memory to shared memory. Since consecutive threads read consecutive memory positions, this operation yields coalesced memory reads.

Once the transfer is over, each thread in the warp determines the subset of results it will write to global memory so that writes can be coalesced (i.e., the \textit{i}-th warp thread  will write the ($i \text{ mod } warpSize$)-th positive results). This is achieved by using the \textit{linearScan} function (line \ref{resSetScan}). Here, each thread performs concurrently a linear scan over the query bit vector words in order to take note of the positive results it will write to global memory. This operation can be carried on efficiently thanks to the shared memories \textit{broadcast} capability, which avoids bank conflicts between threads in a warp if they are all reading the same address. 
Moreover, each thread can store the information related to the positive results it has to write in its own private registers, since the decoding between lines \ref{cycleWarp} and \ref{endCycleWarp} is scalable with respect to the subqueries bit vectors size (which is fixed for a given bitmap).

Once these information are determined, the warp threads finally perform a collective write of the subquery results, yielding coalesced writes (function \textit{writePID}, line \ref{resWrite}). 
Each warp also knows exactly where the results of each subquery has to be stored in global memory, since the overall amount of results per subquery can be determined during the filtering phase and stored in a vector (therefore, an exclusive prefix sum over the vector returns the correct memory location for each result).

\subparagraph*{Decoding -- Complexity.}
Considered we have one bitmap for each active cell $c$,  the overall decoding complexity (Algorithm \ref{lst:decoding}) is:
\begin{equation}
\label{eq:eqDecComp2}
O(\sum_{c \in \cal C_A}|\hat{Q}_I^c|\cdot|P^c| + \sum_{q' \in \hat{Q}_I^c} |R_{q'}|),
\end{equation}
where $R_{q}$ denotes the result set of a query $q$. The first term is due to the access for each cell $c\in {\cal C}_A$ to the respective bitmap, each containing $|\hat{Q}_I^c| \cdot |P^c|$ bits (number of intersecting subqueries times the number of objects in the cell), while the second term is due to the writes of the intersecting subqueries results. As we highlighted for filtering, the grid cells size indirectly affects the decoding complexity.

\subsection{Optimizations}
\label{sec:optimizations}

\subparagraph*{Covering subqueries optimization -- covering subqueries information notification.}
\SQ, \SQB\ and \QQ\ take advantage of the covering subqueries (we denote their set by $\hat{Q}_C$) in order to reduce the result set the GPU computes, thus saving a relevant amount of GPU computations during the filtering and decoding phases and I/O traffic between the GPU and the CPU during the decoding phase. This is achieved by notifying the CPU the covering subqueries data, together with the object locations enclosed in the $\cal C$ cells, just after the end of the \textit{sorting} phase.

\subparagraph*{Covering subqueries optimization -- covering subqueries result set expansion.}
As soon as the data relevant for reconstructing the $\hat{Q}_C$ result set is notified, the CPU can start its expansion. Indeed, for each $q \in \hat{Q}_C$ we have to consider pairs $(q, c \in {\cal C}_\alpha)$, where each \emph{c} represents the index of a cell entirely covered by \textit{q}. After the final sorting step of the indexing phase, the lists of object locations associated with each cell of ${\cal C}$ are sent to the host memory, so that the CPU can directly access them. Therefore, by looking at these lists, the CPU can immediately extract the result set related to $q$. This operation is performed by the CPU in background, between the end of the \textit{sorting} phase and the end of the \textit{decoding} phase of the GPU.

\subparagraph*{Covering subqueries optimization -- complexity.} The cost related to the covering subqueries result set expansion is $O(\sum_{q \in \hat{Q}_C} |R_{q}|)$,  due to the scan of the list of object locations associated with the cell of each covering subquery. Since this task is carried out by the CPU and overlapped with the tasks performed on GPU, 
in practice it has very little or negligible impacts on the overall execution time.

\subparagraph*{Task scheduling optimization.} 
A proper GPU task scheduling policy can substantially improve the overall execution time by reducing the inactivity time of the GPU streaming multiprocessors in presence of unbalanced workload distributions. In this context we define a single workload GPU \textit{task} as the \textit{set of computations} related to the intersecting subqueries falling inside a specific grid active cell, whereas the \textit{computational weight} of a given task is the amount of containment tests associated with it, i.e., the product between the amount of intersecting subqueries and the amount of object locations falling inside the related active cell. 
%

In general we can take into consideration three high-level GPU task scheduling strategies when assigning the workload tasks to the GPU streaming multiprocessors \cite{cederman2008dynamic}, namely, the static task list strategy (which is the default one used by CUDA), the task queue based strategy and the task stealing strategy.

%

Since the computational weight of each task is known a-priori once the sorting phase is performed, and that new sub-tasks cannot be created at run-time, we deem that the first strategy, together with a reordering of the task list according to the tasks computational weight, is the best one for the scenarios considered; indeed, this strategy has the effect of batching together the execution of tasks having similar computational costs at the negligible cost related to the need of accessing the task list atomically whenever an idling GPU streaming multiprocessor available to take in charge the first non-assigned task (atomic access is required to ensure a single execution for each task).

This optimization is used during the filtering and decoding phases when assigning tasks (active $\cal C$ cells) to streaming multiprocessors.

\hide
{
\section{Computational Complexity}
\label{sec:costModel}

In this section we express analytically the (per tick) complexity of the filtering and decoding phases. We focus on these phases since the complexities of the algorithms underlying the preceding ones (\textit{index creation}, \textit{moving objects and queries indexing} and \textit{sorting}) are already reported in Section \ref{sec:commonParts}.

We recall that $P$, $Q$, and $R$ denote the up-to-date moving object positions, the 
non-obsolete queries issued by the objects, and the complete result set associated 
with a generic tick, respectively.
Since only active cells ${\cal C}_A$ have to be
considered when computing the results, we denote by $P^c$, $Q^c$, and $R^c$ respectively, the 
objects, the queries, and the results associated with an active cell $c
\in \cal {\cal C}_A \subseteq C$.
In addition, since from $Q$ two sets of covering/intersecting subqueries $\hat{Q}_C$ and $\hat{Q}_I$
are created, we denote by $\hat{Q}_C^c$ and $\hat{Q}_I^c$ respectively the 
covering and intersecting parts of queries associated with each active cell $c$.

\subsection{Filtering, decoding and covering subqueries result set expansion cost analysis}
\label{subsec:sqCost}

For each active cell $c \in C_A \subseteq {\cal C}$, the filtering phase has to compute the containment test outcomes between the intersecting subqueries and the object falling in $c$. Therefore the associated complexity can be expressed as:
\begin{equation}
\label{eq:eqFiltComp2}
\sum_{c \in \cal C_A} |\hat{Q}_I^c| \cdot |P^c|.
\end{equation}

The subsequent step consists in extracting the positive containment test outcomes from the bitmaps. Again, considered that we have one bitmap for each active cell, the overall complexity of the decoding phase is:
\begin{equation}
\label{eq:eqDecComp2}
\sum_{c \in \cal C_A}|\hat{Q}_I^c|\cdot|P^c| + \sum_{q' \in \hat{Q}_I^c} |R_{q'}|,
\end{equation}

The first term is due to the access for each cell $c\in {\cal C}_A$ to the respective bitmaps, each one containing $|P^c| \cdot |\hat{Q}_I^c|$ bits, while the second term is due to the writes of the intersecting subqueries results.

Finally, the cost related to the covering subqueries result set expansion described in Section \ref{sec:optimizations} is equal to

\begin{equation}
\label{eq:eqCovComp}
\sum_{q \in \hat{Q}_C} |R_{q}|,
\end{equation} 

which is due to the scan of the list of object locations associated with the cell of each covering subquery. As mentioned before, this task is brought on by the CPU and overlapped with the tasks performed on GPU between the end of the \textit{sorting} phase and the end of the \textit{decoding} phase. As a consequence, it usually has very little or negligible impacts on the overall execution time. 

In general we note how decreasing the grid cells size entails a smaller number of containment tests, even if the number of intersecting subqueries to manage is larger. 
However, an arbitrary decrease has negative side effects, such as the fragmentation of the intermediate results in a large number of small bitmaps, which in turn influences the overall running time negatively due to inefficient computational resource usage and scattered memory access. Moreover, this may produce a considerable increase in the amount of intersecting subqueries, which in turn increases the complexities described by Equations \ref{eq:eqFiltComp2} and \ref{eq:eqDecComp2}. Indeed, we argue there is a trade-off between decreasing the overall number of operations executed and optimizing parallelism and memory access costs.

\label{subsec:qqCost}

\subsection{Cost for Uniform  Spatial Distributions}
The equations described in Section \ref{subsec:sqCost} are independent with respect to objects distribution properties, so they apply to all possible scenarios. However, skewed distributions affect negatively the effectiveness of parallelization, since they entail unbalanced GPU resources usage, and it would be interesting to formalize these effects. 

Unfortunately, it is quite difficult to capture this potential loss of parallelism analytically, and it is even more difficult to do this if the spatial index considered is not based on uniform grids; thus, we restrict our analysis to scenarios in which objects are distributed uniformly and issue fixed-sized squared range queries centered on their location. Even if the resulting analysis is limited, it nonetheless allow to derive some interesting considerations related to the dimensioning of grid cells. 

Our analysis will be based on the following parameters: the object spatial density $\rho$, expressed as objects per areal unit, the query rate $qr$, expressed as the probability that an object issues a query during a time tick, the query area $qa$, the grid cell area $ca$ and the total number of cells $nc$.

From these parameters we can derive the overall number of objects locations, $|P| = \rho \cdot ca \cdot nc$, and the expected amount of queries per tick, $|Q|=|P|\cdot qr = \rho \cdot ca\cdot nc\cdot qr $.
For what relates to queries, however, we need to take into consideration the way they are indexed (Definition \ref{def:objPosMap}) and thus discern between intersecting and covering subqueries: given that we are limiting our analysis to spatial indices based on uniform grids (like the one used by \SQ), it is possible to derive the amount of intersecting and covering subqueries from the geometric properties of the originating queries. Indeed, given a set of queries and a grid cell size we can estimate these two quantities as:
\begin{equation}
\label{eq:nr_isect}
\begin{split}
|\hat{Q}_I| \approx \frac{\sum_{q \in Q} perimeter(q)}{cell\_side},\\
|\hat{Q}_C| \approx \frac{\sum_{q \in Q} area(q)}{cell\_side^2} - |\hat{Q}_I|.
\end{split}
\end{equation}

The first formula counts the cells along the query perimeters while the second derives from the observation that the covering subqueries are those corresponding to cells that are part of the area of some query but does not intersect the query perimeter. We observe that Equation \ref{eq:nr_isect} is not limited to squared range queries or uniform spatial distributions when considering uniform grids (like \SQ), while for \QQ\ is applicable only when the materialized quadtree corresponds to a uniform grid (i.e., ${\cal C} = {\cal C}^{l_{deep}}$).

We now proceed determining the costs related to the filtering and decoding phases. Remembering that these phases process only intersecting subqueries, their costs can be computed as the expected cost of processing a single grid cell multiplied by the overall number of cells, \textit{nc}. If the average number of object locations per cell is $\rho \cdot ca$, the intersecting subqueries spatial density is $\rho_{\hat{Q}_I}$, expressed as the amount of intersecting subqueries per areal unit, and the average number of intersecting queries per cell is $qr \cdot \rho_{\hat{Q}_I} \cdot ca$, the filtering cost can be directly derived from Equation \ref{eq:eqFiltComp2} as:
\begin{equation} 
\label{eq:filter2}
O\Bigg(nc \cdot (\rho \cdot ca) \cdot (qr \cdot \rho_{\hat{Q}_I} \cdot ca)\Bigg)
\end{equation}

For what relates to the decoding phase,
we have that if the average number of results per intersecting subquery is $qa_{\hat{Q}_I} \cdot \rho$, where $qa_{\hat{Q}_I}$ is the intersecting subqueries average area, from Equation \ref{eq:eqDecComp2} we have:
\begin{equation}
\label{eq:decode2}
O\Bigg( nc \cdot (\rho \cdot ca) \cdot (qr \cdot(\rho_{\hat{Q}_I} \cdot ca))\, +\, nc \cdot (qr \cdot \rho_{\hat{Q}_I} \cdot ca) \cdot (\rho \cdot qa_{\hat{Q}_I}) \Bigg)
\end{equation} 

Both for filtering and decoding we observe that the finer the grid is, the higher is $\rho_{\hat{Q}_I}$ and the lower is $\rho$. Therefore, both filtering and decoding have individual minima obtained by finding a cell size entailing the best tradeoff between $\rho_{\hat{Q}_I}$ and $\rho$. This in turn requires to find a cell size minimizing the combined filtering and decoding cost.

Finally, we saw how the cost related to the covering subqueries result set expansion is decoupled from the filtering and decoding phases, since it is distributed among various pipeline phases and performed on CPU in background. Its overall cost is due to the scan of the object locations inside the cells covered by the covering subqueries and can be expressed as:

\begin{equation}
\label{eq:covering}
O\Bigg(nc \cdot (qr \cdot \rho_{\hat{Q}_C} \cdot ca) \cdot (\rho \cdot qa_{\hat{Q}_C}) \Bigg),
\end{equation} 

where $\rho_{\hat{Q}_C}$ is the amount of covering subqueries per aerial unit while $qa_{\hat{Q}_C} = ca$ represents the covering subqueries area.

}

\section{Experimental setup}
\label{sec:setup}

All the experiments are conducted on a PC equipped with an Intel Core i3 560 CPU, running at
3,2 GHz, with 4 GB RAM, and an Nvidia GTX 560 GPU with 1 GB of RAM coupled with CUDA 5.5.
The OS is Ubuntu 12.04.
We exploit a publicly available framework \cite{sow2012software, sowell2013experimental} for both workload generation and testing. 
The framework comes with a number of sequential, CPU-based iterated spatial join algorithms. Among these, the \emph{Synchronous Trasversal} algorithm (\CPU) is shown to be consistently the best~\cite{sowell2013experimental} and thus we compare our GPU-based approaches with this algorithm.
%

As regards our GPU-based proposed solutions (\QQ\ and \SQ), we slightly modified the framework in order to offload the most time-consuming parallelizable tasks to the GPU, while delegating the others (mostly related to the GPU management) to the CPU.

We use three types of synthetic datasets: \emph{(i)} \textit{uniform datasets}, in which moving objects are distributed uniformly in the space; \emph{(ii)} \textit{gaussian datasets}, in which moving objects tend to gather around multiple \textit{hotspots} by following a normal distribution. The skewness in the gaussian datasets depends on the number of hotspots: the more the hotspots are, the more the objects tend to be uniformly distributed; \emph{(iii)} \textit{network datasets}, in which moving objects are distributed uniformly over the edges of a bidirectional graph representing a road network. In our experiments we use the San Francisco road network, derived from TIGER/Line files. This kind of datasets are characterized by a mild skewness, due to the constraint on the positioning of the objects. 
All the datasets are created using the generator provided by the framework, which is partly derived from the Brinkoff generator \cite{brinkhoff2002generator}.

In all tests we compute repeated range queries over 30 ticks. To model object movements the framework generates 30 instances of each dataset, one for each tick.
%

%
Table \ref{tab:param} summarizes the main parameters used to generate the datasets. The listed parameters apply to all the datasets, except for the amount of hotspots which is relevant for gaussian datasets only. The framework uses a generic spatial distance unit $u$ (e.g., meters).

\begin{table}[htdp]
\centering
\small
\begin{tabular}{|c|p{.65\textwidth}|}
  \hline\hline
  \emph{Spatial region} & 
  All tests occur in a squared spatial region with side length of $22500\ u$.
  \\ \hline
  \emph{Amount of objects} & 
  We vary the number of moving objects from $100K$ to $1500K$. In some tests the number of moving objects is fixed and the exact amount is explicitly stated in their descriptions.
  \\ \hline
  \emph{Objects maximum speed} &
  In all tests the maximum speed of each object is fixed to $200\ u$ per tick ($\Delta t$), where the objects are allowed to change their speed \cite{sowell2013experimental}. In general, changes in speed may slightly alter the objects distribution but do not change the distribution general properties.
  \\ \hline
  \emph{Query rate} &  The percentage of objects that issue a range query during every tick is always set to 100\%.
  \\ \hline
  \emph{Query size} & 
  All queries in a test are squared and, depending on the experiment, they may be all equally sized or not. We vary the side length in the range $[200\ u, 800\ u]$.
  The default value is $200\ u$.
  \\ \hline  
  \emph{Amount of ticks} & 
Whenever not specified, the default amount of ticks, corresponding to different snapshots of a dataset, is 30. Consecutive snapshots are expected to exhibit slight changes, according to the properties of the dataset spatial distribution.
  \\ \hline  
  \emph{Query location} & 
  All the queries 
  are centered around the objects issuing them.
  \\ \hline  
  \emph{Amount of hotspots} & 
  Depending on the experiments goals and specificities, the amount of hotspots is varied in the $[10, 150]$ range. Whenever not specified the default value used is $25$.
  \\ \hline  \hline
\end{tabular}
\caption{\label{tab:param}
Data and workload generation parameters.}
\end{table}



The decoded results are produced by the GPU in blocks, i.e., for each query/subquery a list of (positive) results is produced, whereas for \CPU\ the results are produced one by one. To avoid bias in the performance comparison, we thus force \QQ\ and \SQ\ to report the GPU-generated results to the framework in pairs, hence expanding the lists of objects belonging to the result set of each query.  


\section{Experimental evaluation}
\label{sec:experiments}

%
The experimental studies conducted for this work are introduced below, and are denoted by $S1,\ldots,S8$:

\begin{itemize}
\item[\emph{S1}] We study how a \emph{lock-free data structure}, like the bitmap
proposed to encode the intermediate output of the range queries, entails considerable improvements over a baseline GPU algorithm that recurs to locks to assure the result buffer consistency. 

\item[\emph{S2}] We analyze the advantages coming from the \emph{covering subquery optimization} in reducing computations (and related I/O traffic) performed on the GPU side.

\item[\emph{S3}] We show how a proper GPU \emph{task/block scheduling} can improve the \SQ\ and \QQ\ performances by reducing the workload unbalances deriving from skewed spatial distributions.

\item[\emph{S4}] We study how the data distribution skewness influences the choice of the optimal grid coarseness, focusing on \SQ.

\item[\emph{S5}] We study how \QQ\ is able to automatically adapt to the spatial data distribution of a dataset, even when the distribution is highly skewed.

\item[\emph{S6}] We analyze the impact of various spatial distributions on the performance. To this regard we study mean and dispersion index related to the amount of objects per active cell achieved by \SQ\ and \QQ, and the relationships that these measures have with some important features that impact the performance of the system, such as the overall amount of subqueries and the proportion of covering/intersecting ones.
%
%

\item[\emph{S7}] This study analyzes the sensitivity of the \SQ\ and \QQ\ performances with respect to datasets characterized by different spatial distribution properties, such as the amount of objects and the query area.

\item[\emph{S8}] This final study analyzes how the main factors characterizing the datasets, such as the amount of objects, the query rate, the query area and the skewness, affect the system bandwidth $\beta$ (as defined in Section \ref{sec: qos considerations}).
\end{itemize}

\noindent It is worth remarking that in the final \emph{S6}...\emph{S8} studies we turn on all the optimizations devised for both \SQ\ and \QQ. In particular, a relevant amount of computations is avoided by distinguishing between covering and intersecting subqueries; we exploit the lock-free bitmaps to store the intersecting subqueries intermediate results; we heuristically balance the workload between the GPU SMs by reordering the tasks/blocks to be scheduled, from the heaviest to the most lightweight; finally, we always adopt the best possible grid coarseness for \SQ, given any dataset, by means of an oracle, although this strategy cannot be adopted in practical settings since it requires additional work to profile the \SQ\ performance for the various datasets.



\subsection{Analysis on the benefits coming from the usage of bitmaps (S1)}

In this study we evaluate the benefits coming from the usage of a lock-free data structure such as the bitmaps proposed. We focus on the filtering and decoding steps, since these are the only phases during 
which the bitmaps are utilized in order to significantly improve the performances. 
We limit this comparison to \SQ\ and \SQB, since \QQ\ filtering and decoding phases are similar to \SQ\ ones, and the benefits over a naive lock-based technique are thus analogous.

We briefly remember that in \SQB\ the query result set is computed and transmitted to the CPU counterpart on the fly, during the filtering step, thus avoiding the need of a subsequent decoding phase. 
Moreover, since the results of different queries can be interleaved in the output stream, 
each result is represented as a pair (\emph{query:point}).

\begin{figure}[h!]
\centering
    \includegraphics[width=.5\columnwidth]{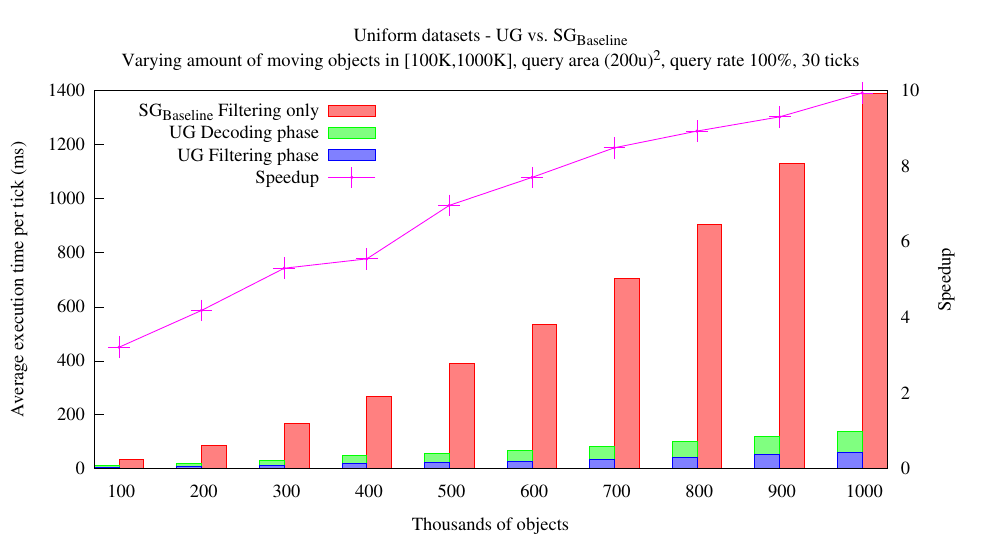}
    \caption{\SQ\ vs. \SQB\ time analysis for uniform datasets by varying the number of objects in [100K,1000K]. 
The histograms of the \SQ\ filtering and decoding phases are time-stacked for clarity purposes.}
    \label{fig:plot_SQ_vs_SQB_Challenge}
\end{figure}


In Figure \ref{fig:plot_SQ_vs_SQB_Challenge}  we can see how \SQ\ outperforms \SQB, even when considering modest workloads. 
For our purposes, it is enough to compare \SQ\ and \SQB\ on a uniform spatial distribution of object/queries.
Note the line representing the speedup obtained by \SQ\ over \SQB\ per single experiment: the \SQB\ performance gets worse when the number of objects (and thus the output size) increases. This indicates that the main bottleneck of \SQB\ is the synchronization mechanisms adopted, which affect negatively the performance when the amount of results increases.

\subsection{Covering subqueries optimization (S2)}
\label{sec:coveringStudy}

The overall goal of this study is to analyze the benefits coming from the covering subqueries optimization described in Section \ref{sec:optimizations}.
We briefly remember that this technique aims to speed up the query processing in three ways: 
first, by reducing the overall amount of containment tests performed by the GPU; 
second, by reducing the amount of results determined at the GPU side, and thus the amount of data the GPU has to send back to the CPU once the filtering and the decoding phases are over; third, by leaving the CPU in charge of expanding the covering subqueries result sets, 
during the same time that the GPU processes the intersecting subqueries.

We focus on \QQ, given its more advanced spatial indexing and considered that, in terms of query covering 
management, \SQ\ carries out the same operations per each subquery covering an 
index cell. We compare two different versions of \QQ : the former, denoted by \textsf{Covering ON}, is the version 
where the covering subquery optimization is exploited. The latter, denoted by \textsf{Covering OFF}, does not exploit the knowledge about the covering queries, thus considering all the subqueries as intersecting.

In our experiments we want to independently focus on two key parameters: spatial distribution \textit{skewness} and \textit{query area}.
The skewness consistently influences \emph{(i)} the \textit{ratio} between covering and intersecting subqueries, and \emph{(ii)} the \textit{weight} of the covering subqueries in terms of the percentage of generated results.
As regards \emph{(i)}, the more the objects tend to gather in specific places, the more \QQ\ refines the grid in those areas, in turn increasing the aforementioned ratio. As regards \emph{(ii)}, the smaller the average  size of the \QQ\ cells is, the higher the probability that a subquery area completely covers a grid cell. Note that \QQ\ materializes dynamically the index cells, and generate smaller cells in correspondence to the spatial regions with higher object density. On these small and highly populated cells the ratio of covering subqueries to intersecting ones gets larger. 
This in turn increases the overall amount of results obtained from the covering subqueries, with positive returns on the performance.
%

Query area is another important factor, because it directly influences the ratio between covering and intersecting subqueries.
Hence, we want to analyze how much this parameter influences the performances, aside from the skewness.

\begin{figure}[h!]
\includegraphics[width=\columnwidth]{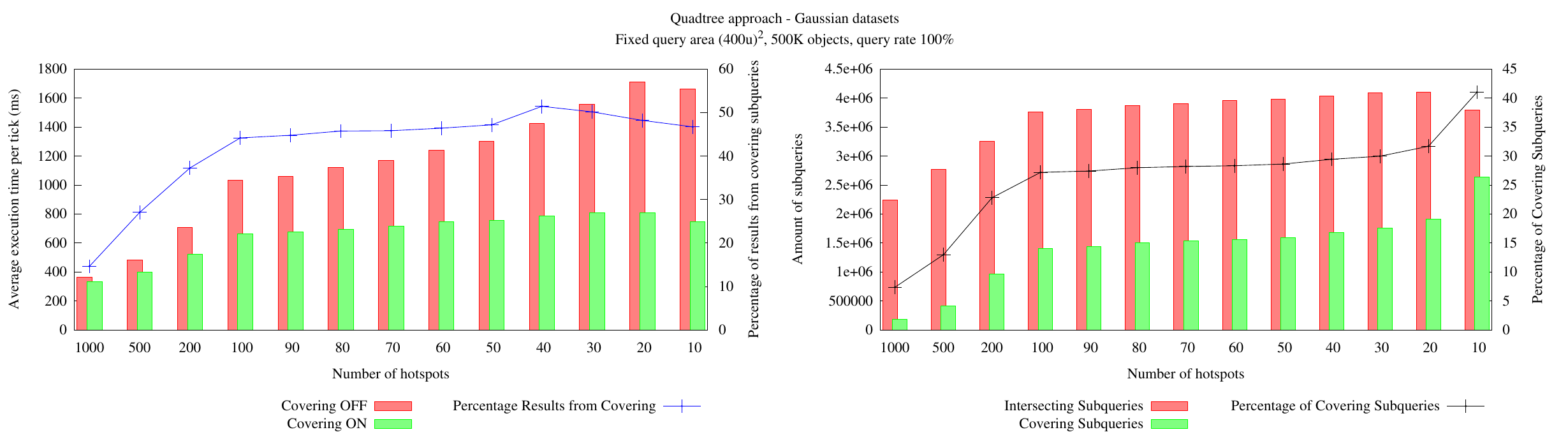} 
\caption{Gaussian datasets, 500K objects, query area $(400u)^2$, varying hotspots in [50,1000], query rate 100\%, Covering ON vs OFF.}
\label{fig:plot_Gaussian_CoveringChallenge}
\end{figure}

In the first batch of experiments we  vary
the skewness of the object spatial distribution in a set of gaussian datasets, by changing the number 
of hotspots in the interval [10,1000], while the amount of objects and the query area are kept fixed at 500K and $(400u)^2$, respectively.
The left plot of Figure \ref{fig:plot_Gaussian_CoveringChallenge} 
shows how the exploitation of the covering subqueries greatly reduces the overall execution time, 
in particular when the skewness increases by reducing the amount of hotspots.
The more the skewness is, the more the percentage of results coming from the covering subqueries is, thus increasing the performance gap between the two versions when the skewness gets larger.
Finally, the right plot of Figure \ref{fig:plot_Gaussian_CoveringChallenge} gives a further explanation of the observed behaviour, and shows that the skewness is directly proportional to (or equivalently, the number of hotspots 
is inversely proportional to) the covering/intersecting ratio.
%

\begin{figure}[h!]
\includegraphics[width=\columnwidth]{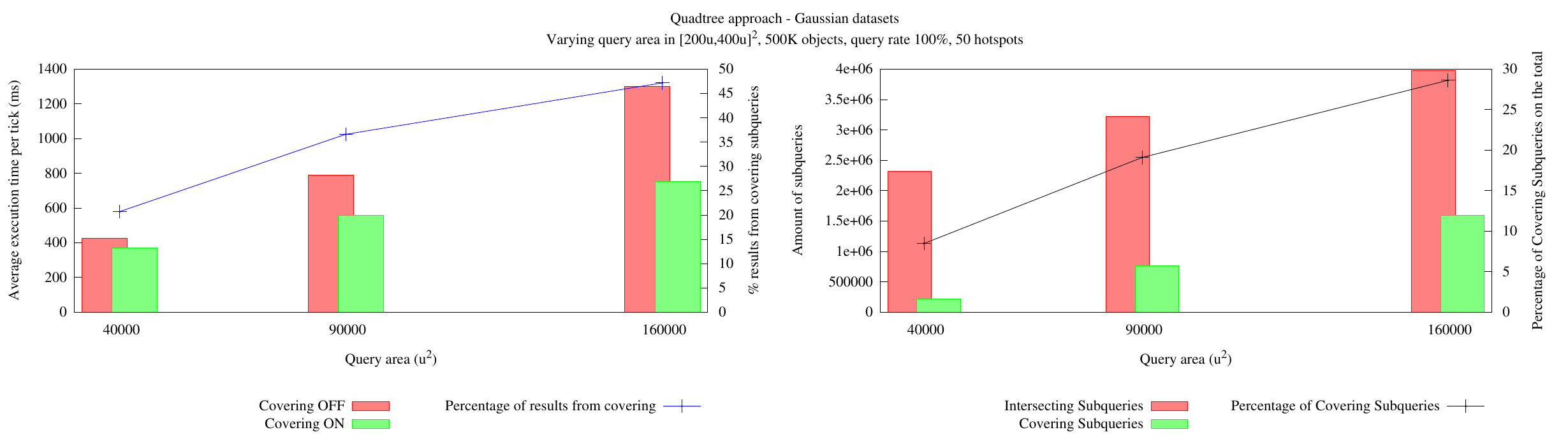} 
\caption{Gaussian datasets with 50 hotspots, 500K objects, query area varied in $[(200u)^2,(400u)^2]$, query rate 100\%, Covering ON vs. OFF.}
\label{fig:plot_GaussianVQA_CoveringChallenge}
\end{figure}

In the second batch of experiments we focus on a gaussian dataset having a fixed amount of 50 hotspots,  thus characterized by a 
moderately skewed distribution. We vary the query area in the $[(200u)^2,(400u)^2]$ range, while 
the amount of objects is kept fixed at 500K.
Figure \ref{fig:plot_GaussianVQA_CoveringChallenge} shows that, when the query area is increased,  
the gap between the two versions of \QQ\ gets larger as well, 
while the covering/intersecting subqueries ratio strictly follows the trend.

In conclusion, we argue that the percentage of results coming from the covering subqueries determines the extent of the advantages possibly coming from this optimization.
This figure is essentially determined by the skewness characterizing the spatial object distribution,
while the query area amplifies or reduces this phenomenon by changing the query results redistribution ratio between the covering and the intersecting subqueries.

\subsection{Task scheduling policy (S3)}
\label{sec:expTaskScheduling}

In this section we investigate how the task scheduling optimization described in Section \ref{sec:optimizations} can substantially improve the overall execution time of \SQ\ and \QQ\ by redistributing more evenly the workload among the GPU streaming multiprocessors.
Besides analyzing the execution times in this study we also collect and study profiling data concerning the containment tests actually carried on by every streaming multiprocessor.


\begin{figure}[ht]
\includegraphics[width=\columnwidth]{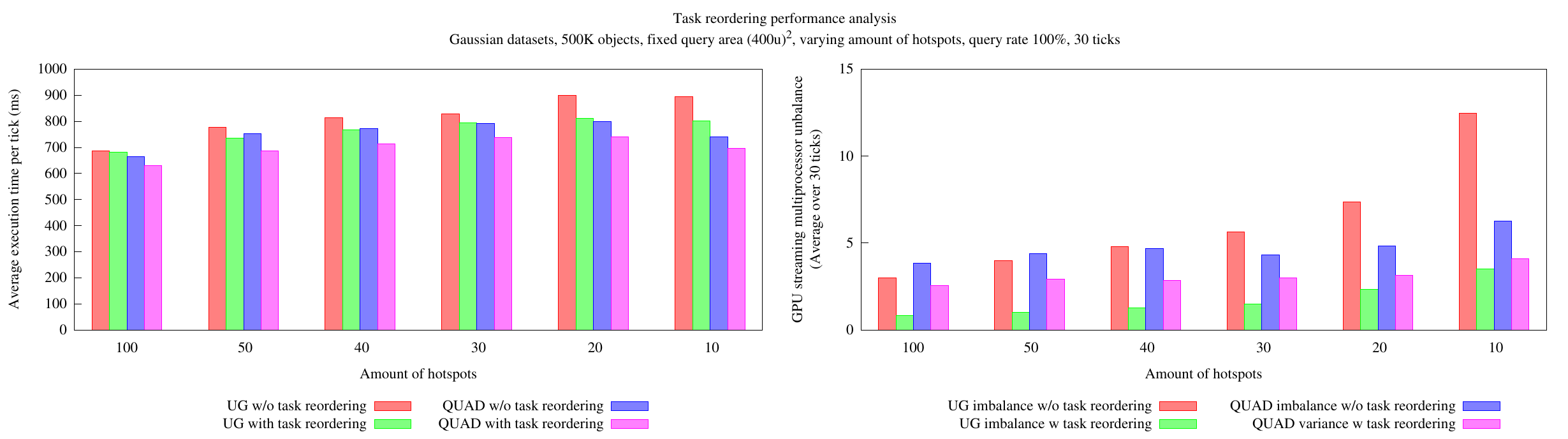} 
\caption{Analysis on the performances and workload redistribution among the GPU streaming multiprocessor with and without the static task list reordering - gaussian datasets, 500K objects, query area $(400u)^2$, query rate 100\%, varying amount of hotspots. The left plot refers to the execution times observed while the right one refers to the profiling data collected during the filtering phase (decoding phase data is analogous).}
\label{fig:plot_gaussian_scheduling}
\end{figure}

The left plot in Figure \ref{fig:plot_gaussian_scheduling} shows that the reordering always reduces the execution times of \SQ\ and \QQ. 
Note that in the case of \SQ\ the reordering entails higher performance improvements when the skewness gets large, while \QQ\ improvements are always moderate due to the ability of its underlying spatial indexing to dynamically produce tasks/blocks of similar weights. 
%

We study in depth this behaviour by profiling the execution of the GPU SMs. In particular, we collect the per-tick amount of containment tests performed by each SM, and check whether the observed performance trends are reflected in workload unbalances among the SMs. 
In this context we define the \textit{SM imbalance} measure during a single time tick as the \textit{relative difference} between the highest amount of containment tests performed by a single GPU streaming multiprocessor with the lowest amount performed by an another SM. Then, we compute the average of this measure across the ticks in order to characterize the average workload unbalance. From the right plot in Figure \ref{fig:plot_gaussian_scheduling}, we see how the trend of the \textit{SM imbalance} follows the trend of the execution time, observed in the left plot of the same figure for both \QQ\ and  \SQ.
Hereinafter, all the experiments will be conducted by using the task list reordering optimization.

\subsection{Data skewness and optimal grid coarseness for \SQ\ (S4)}
\label{sec:expSkew}

The following set of experiments aims to show that the best coarseness used for the uniform grid onto which the \SQ\ spatial indexing relies depends on the specificities of the spatial distribution characterizing the objects at each tick.
We therefore aim to show how it is not possible to find a unique optimal MBR \emph{split factor} (i.e., the number of columns/rows in which the MBR is decomposed) that holds for all the datasets. Even more, we show how each pipeline phase has its own optimal MBR split factor given a single dataset.

We first focus on a gaussian dataset characterized by a mild skewness (150 hotspots), and study how the \SQ\ performance 
changes (during a single time tick) by varying the split factor. 
We decompose the overall execution time in three macro phases, namely the \emph{indexing}, \emph{filtering}, and \emph{decoding} phases, where the former includes the \textit{index creation} and the \textit{object/query indexing} phases (Sections \ref{par:indexDetOverviewSQ} and \ref{par:indexDetOverviewQuad}).
%
%
\begin{figure}[!ht]
    \begin{minipage}{0.5\columnwidth}
      \includegraphics[width=\columnwidth]{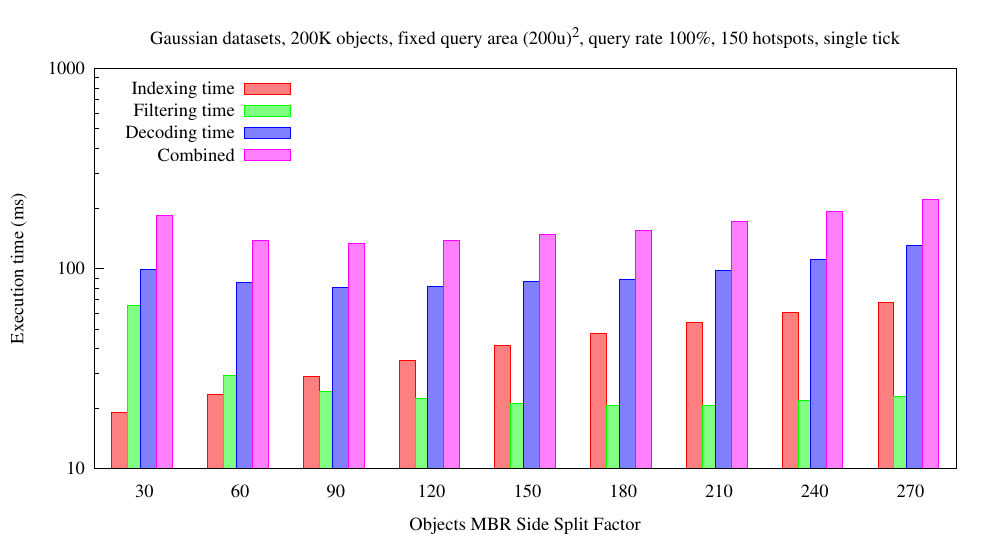} 
\caption{Gaussian dataset, 200K objects, query area $(400u)^2$, query rate 100\%, 150 hotspots. The optimal value is equal to 110. Logscale on the y-axis is conveniently used to magnify small differences in the filtering execution times.}
\label{fig:plot_gaussian1}
    \end{minipage}
    \hspace{1em}
    \begin{minipage}{0.5\columnwidth}
	\includegraphics[width=\columnwidth]{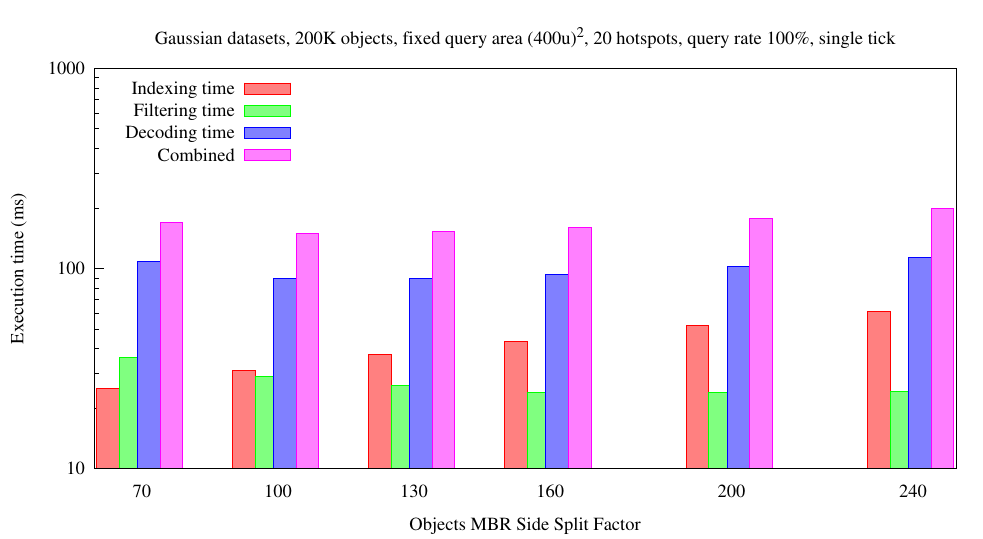} 
\caption{Gaussian dataset, 200K objects, query area $(400u)^2$, query rate 100\%, 20 hotspots. The optimal value is equal to 95. Logscale on the y-axis is conveniently used to magnify small differences in the filtering execution times.}
\label{fig:plot_gaussian2}
    \end{minipage}
\end{figure}
Figure \ref{fig:plot_gaussian1} shows how the \textit{indexing} time gets larger when we increase the MBR side split factor, as expected according to the costs described in Sections \ref{par:indexDetOverviewSQ} and\ref{par:indexDetOverviewQuad}, due to the increase in the amount of subqueries created.
As for the \textit{filtering} phase, we see how the execution times trend exhibit a minimum. 
In general, too small split factors imply very large cells, few or none covering subqueries and potentially large workload unbalances, depending on the skewness. On the other hand, when the split factor is too large too many subqueries may be created, as well as there could be many active cells with small amounts of objects: this may represent a serious pitfall for an efficient usage of the memory/computational resources of a GPU.  
The same reasonings hold for the \textit{decoding} phase as well.
The overall execution time (the \textit{Combined} bar in the plots) has a minimum obtained by using an optimal split factor equal to 110.

%

We replicate the same set of experiments with a consistently skewed gaussian dataset (Figure \ref{fig:plot_gaussian2}). The trends observed in 
Figure \ref{fig:plot_gaussian1} are confirmed, although the \SQ\ optimal split factor value is different (95) due to the different dataset characteristics. This confirms that datasets having different spatial properties require the materialization of grids having different spatial characteristics in order to achieve the best possible performance.
%
%

\subsection{Data skewness and optimal cell size for \QQ\ (S5)}
\label{sec:SkewQQ}

As already described in Section \ref{par:indexDetOverviewQuad}, in \QQ\ the size of the various cells is determined dynamically on the basis of data distribution and according to $th_{quad}$, a threshold determining whether a quadtree quadrant needs to be split at the next level according to the amounts of objects it contains at the time tick the quadtree is computed. Thus, we need to determine an optimal value for $th_{quad}$ which hopefully does not change for datasets characterized by different object spatial distribution or query areas. 
%

\begin{figure}[h]
    \begin{minipage}{0.5\columnwidth}
      \includegraphics[width=\columnwidth]{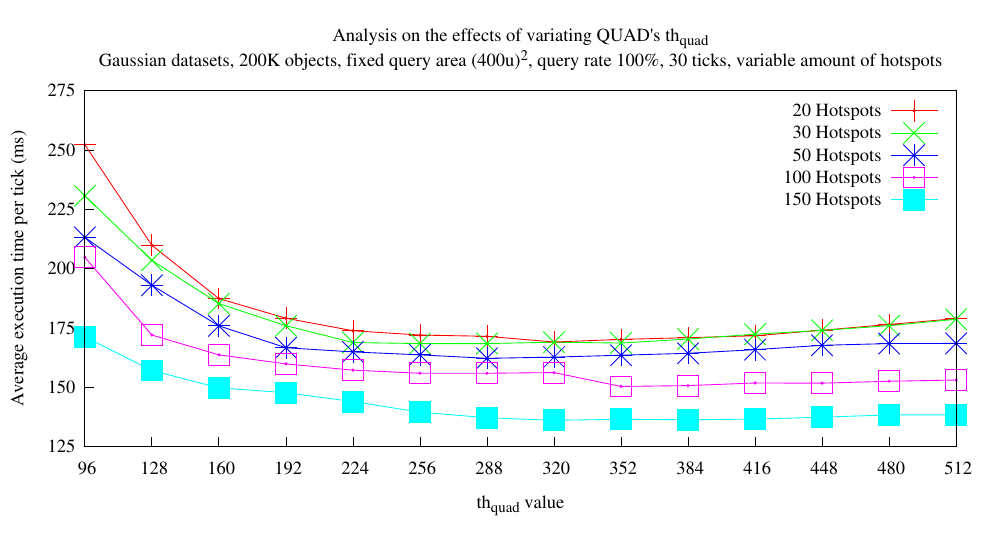} 
\caption{Performance analysis with different \QQ\ $th_{quad}$ values when varying the skewness degree.}
\label{fig:plot_variateQUAD_Hotspots}
    \end{minipage}
    \hspace{1em}
    \begin{minipage}{0.5\columnwidth}
	\includegraphics[width=\columnwidth]{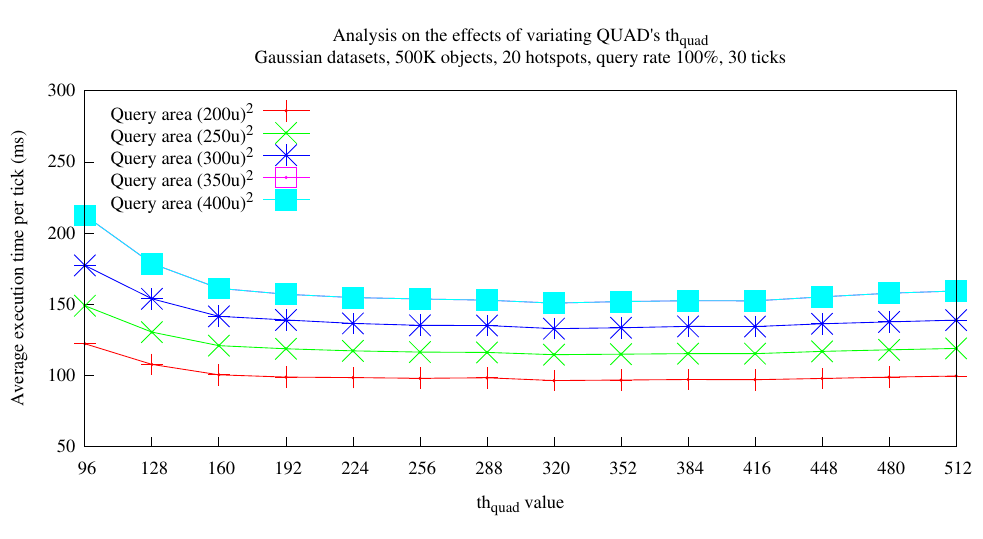} 
\caption{Performance analysis with different \QQ\ $th_{quad}$ values when considering different query areas.}
\label{fig:plot_variateQUAD_VQA}
    \end{minipage}
\end{figure}

Figure \ref{fig:plot_variateQUAD_Hotspots} refers to a set of experiments in which, 
given a set of gaussian datasets with different amounts of hotspots, $th_{quad}$ is varied in order to observe how \QQ\ behaves. 
The amount of objects characterizing each dataset is set to 200K, which allows the exploration of an extensive range of $th_{quad}$ values: lower $th_{quad}$ values increase the amount of resulting subqueries, which in turn increase the amount of GPU memory required for storing the subqueries.
%
 Figure \ref{fig:plot_variateQUAD_VQA} refers to a similar set of experiments in which the query area is varied among the datasets while the other characteristics are kept fixed.

In general we see how \QQ\ is resilient to dataset changes thanks to its low sensitivity with respect to $th_{quad}$, allowing an easy tuning of the system. Moreover, the search for an optimal $th_{quad}$ is not so crucial, given the stability exhibited by \QQ\ for an ample interval of values. Increasing the query area has just the effect of increasing the execution times, while the trend remains the same for all the curves.
%
Considering the results obtained above, in the experiments that follow we set $th_{quad} = 384$.

\subsection{Impact of spatial distribution skewness on the performance (S6)}
%
\label{par:expSkewPart1}

In this study we want to observe how \SQ\ and \QQ\ perform when varying the skewness degree by considering a set of gaussian datasets having different amounts of hotspots. In the experiments that follow we 
keep fixed the amount of objects (500K), the query area ($400u^2$) and the query rate ($100\%$), whereas we vary the amount of hotspots in the $[10,200]$ interval. For \SQ\ and \QQ\ we exploit all the optimizations, including the oracle used by \SQ\ (even though unusable in a practical setting).

Figure \ref{fig:plot_Gaussian_HOTSPOT_Timings} shows that \SQ\ and \QQ\ have similar performances until the skewness becomes consistent, i.e., the amount of hotspots gets below 20. This is confirmed by
the fact that \QQ\ is able to maintain stable and consistent speedups with respect to \CPU, even in presence of extremely skewed distributions, while \SQ\ slightly degrades.

\begin{figure}[h]
    \begin{minipage}{0.5\columnwidth}
      \includegraphics[width=\columnwidth]{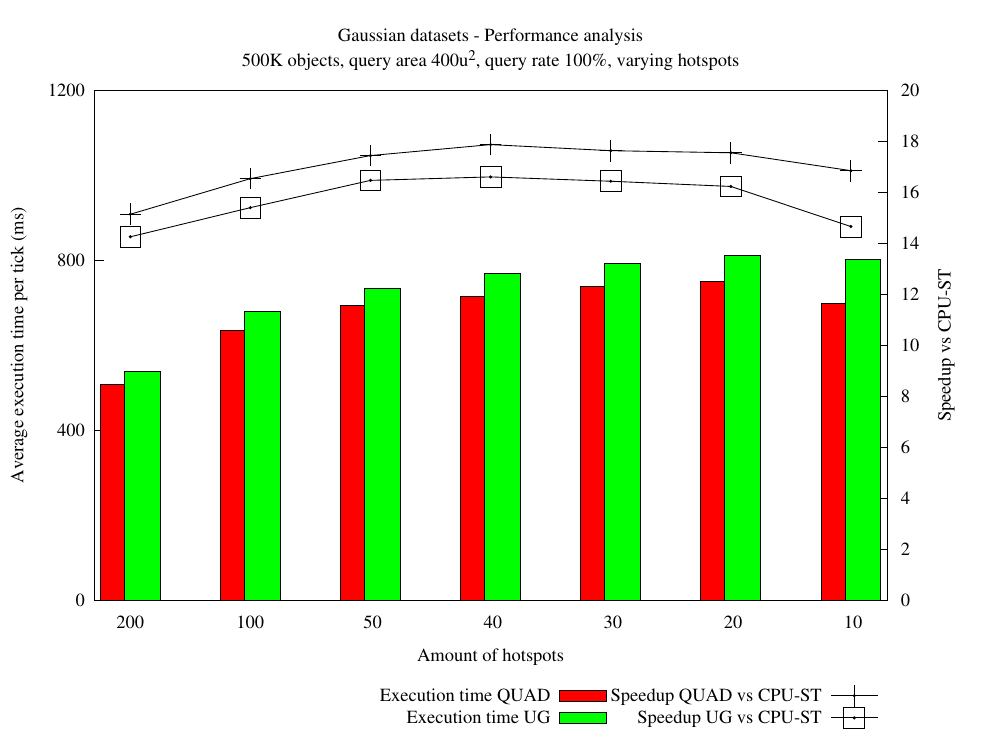}
      \caption{Gaussian datasets, 500K objects, query area $(400u)^2$, amount of hotspots varied in
[10,200], average running times per tick and speedup with respect to \CPU.}
\label{fig:plot_Gaussian_HOTSPOT_Timings}
    \end{minipage}
    \hspace{1em}
    \begin{minipage}{0.5\columnwidth}
	\includegraphics[width=\columnwidth]{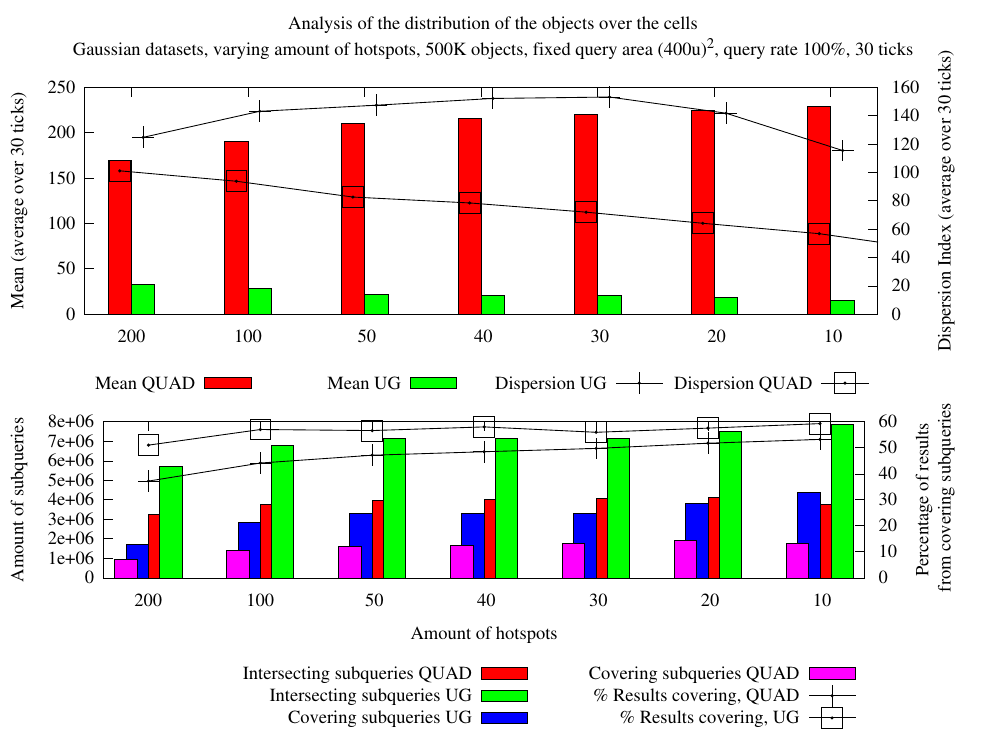} 
	\caption{Gaussian datasets, 500K objects, query area $(400u)^2$, amount of hotspots varied in
[10,200], mean and dispersion index over the grid active cells.}
	\label{fig:plot_Gaussian_HOTSPOT_Stats}
    \end{minipage}
\end{figure}

We try to explain the observed performances in terms of the ability of \SQ\ and \QQ\ in 
redistributing the objects among the grid cells. To this end, we compute the \textit{mean} and \textit{variance} of the amount of objects in each grid active cell and the associated \textit{dispersion index} ($D = \sigma^2 / \mu$) characterizing the distribution of the objects over the active cells.
In general it is expected that, the finer a grid is, the lower the resulting mean and dispersion index are, although these figures are heavily influenced by the skewness characterizing the dataset. 
%
The mean and the dispersion indices obtained by \SQ\ and \QQ\ are shown in the top plot of Figure \ref{fig:plot_Gaussian_HOTSPOT_Stats}. 
\SQ\ always yields remarkably higher dispersion indices and lower means than the \QQ\ ones. 
The very low means observed for \SQ\ depend on the very fine uniform grid exploited, needed to avoid 
heavy populated cells which would entail very expensive tasks to be execute by a single GPU SM.
Indeed, the ability of \QQ\ in properly redistributing workloads associated with objects living in densely populated regions is confirmed by the overall amounts of (intersecting/covering) subqueries produced (Figure \ref{fig:plot_Gaussian_HOTSPOT_Stats}, bottom graph) - amounts which are remarkably lower than the ones obtained by \SQ. 
For example, the amount of subqueries obtained when analyzing a very skewed dataset (such as the gaussian one with 10 hotspots) is about 12 millions for \SQ, and approximatively half for \QQ. Even if the size of the \SQ\ cells is very small, and thus the probability that a subqueries ``covers'' a grid cell gets large, the same plot shows that the proportion of results coming from covering subqueries in \QQ\ are almost on a par with the one obtained by \SQ.
%
Finally, the remarkable smaller count of subqueries allows \QQ\ to have lower GPU memory requirements 
than \SQ\ when generating and computing the subqueries.

\subsection{Performance analysis for different spatial distributions, amount of objects, and query areas (S7)}

In this study we analyze the performances of \SQ\ and \QQ\ with datasets characterized by different spatial distributions. In the experiments that follow we also vary the amount of objects and the query areas. 
For \SQ\ and \QQ\ we exploit all the optimizations.
The goal is to show how \QQ\ is generally able to outperform \SQ, even if the latter relies 
on an expensive, and thus unfeasible, performance profiling (oracle) in order to select the best possible uniform grid coarseness for any dataset. 


\paragraph*{\textbf{Variable amount of moving objects.}}
\label{par:expVE}

\begin{figure}[!h]
    \begin{minipage}{0.5\columnwidth}
      \includegraphics[width=\columnwidth]{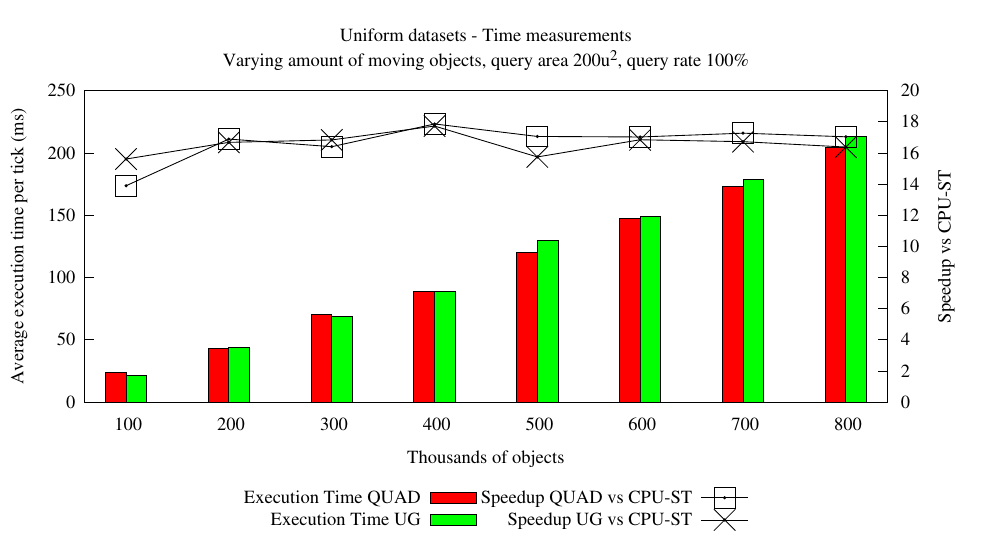}
     \\
     \includegraphics[width=\columnwidth]{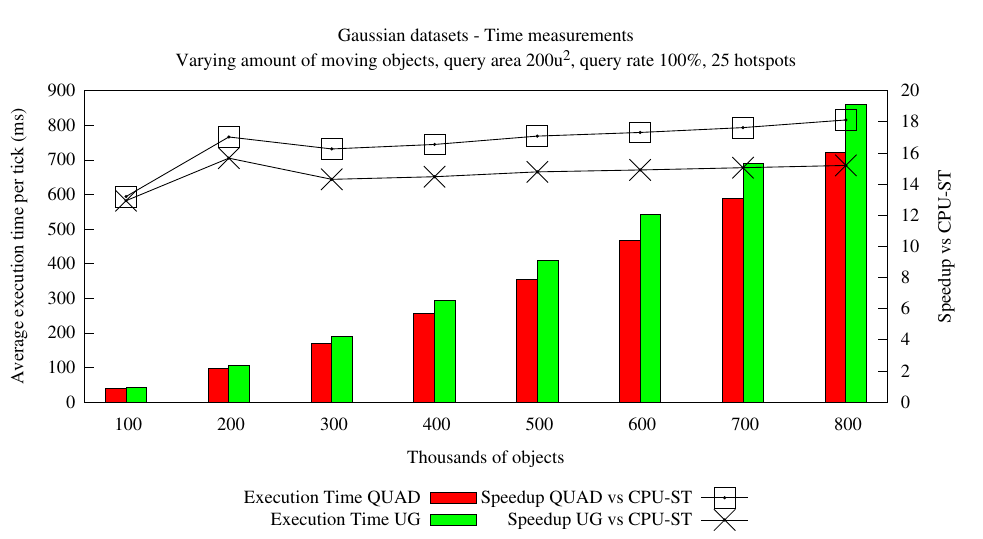}
     \\
     \includegraphics[width=\columnwidth]{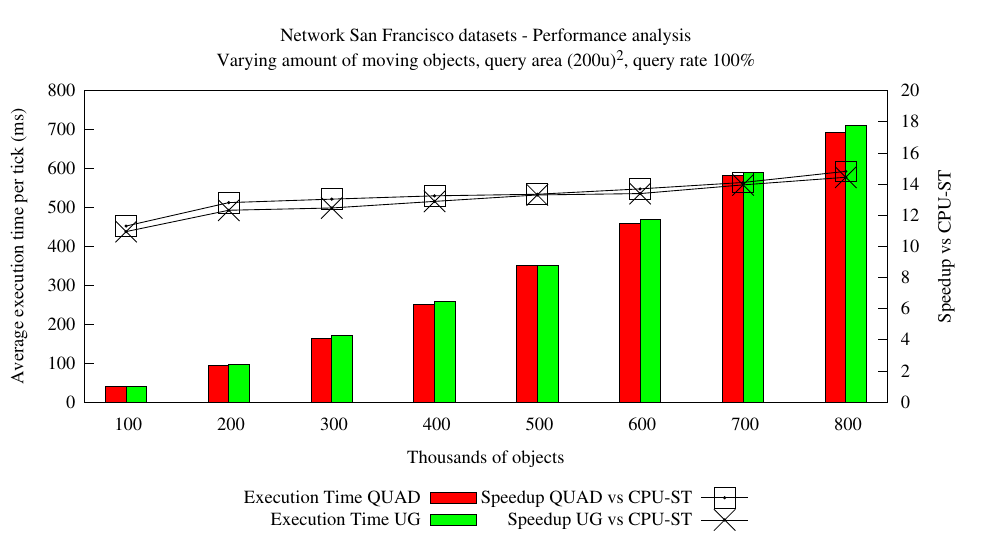}
\caption{Varying the number of objects: average running time per tick and speedup versus \CPU. 
From top to bottom: uniform datasets, gaussian datasets with 25 hotspots, and San Francisco Network datasets. 
}
\label{fig:plot_varying_objects}
    \end{minipage}
    \hspace{0.5em}
    \begin{minipage}{0.5\columnwidth}
    \includegraphics[width=\columnwidth]{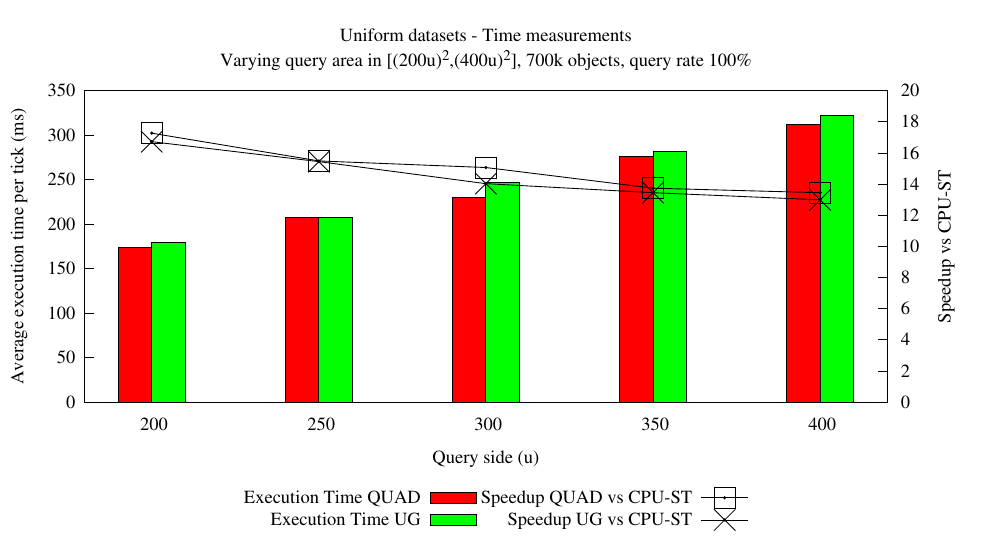}
    \\
    \includegraphics[width=\columnwidth]{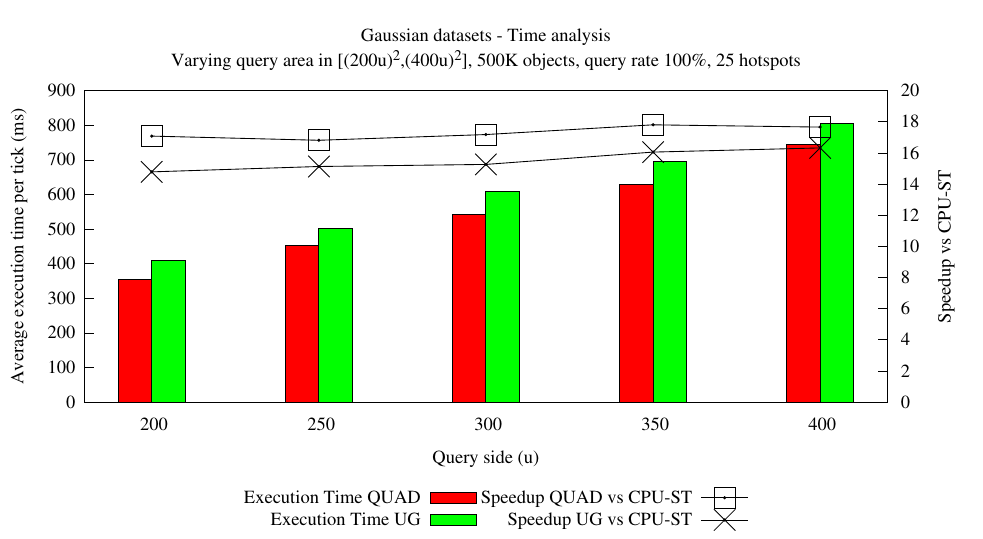}
    \\
    \includegraphics[width=\columnwidth]{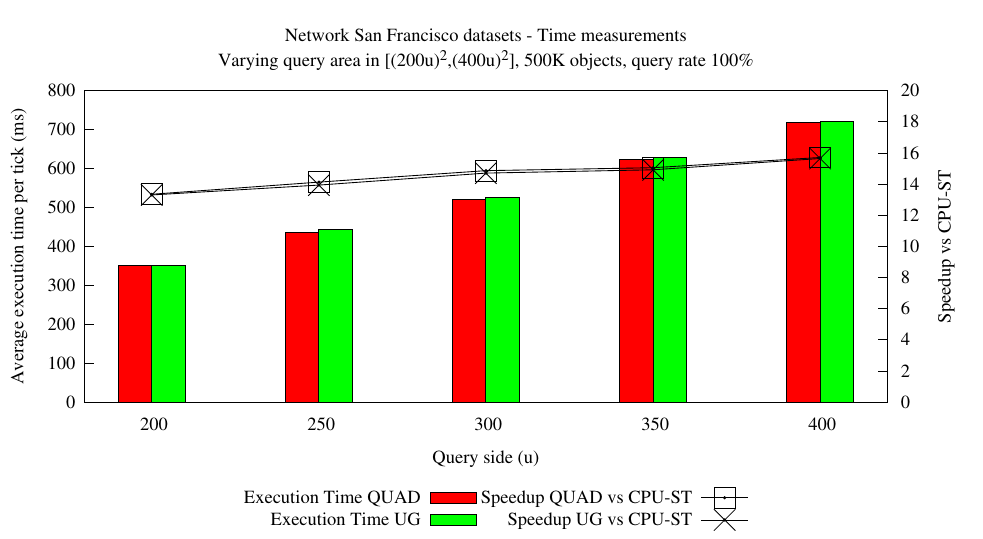}
\caption{Varying the query area: average running time per tick and speedup versus \CPU. 
From top to bottom: uniform datasets, gaussian datasets with 25 hotspots, and San Francisco Network datasets.}
\label{fig:plot_varying_qarea}
    \end{minipage}
\end{figure}

In these experiments we exploit three types of datasets - uniform, gaussian and network-based -
where we keep fixed the query rate and the query area at $100\%$ and $(200u)^2$, respectively.
For the gaussian datasets, the number of hotspots is fixed to 25.
Figure \ref{fig:plot_varying_objects} shows the execution times and the speedups versus \CPU\ 
for these three types of datasets when varying the amount of objects.

When uniform distributions are considered, \SQ\ and \QQ\ exhibit similar performances as expected.
%
When gaussian datasets are considered, 
\SQ\ and \QQ\ exhibit stable and consistent performances, with \QQ\ performing noticeably better than \SQ. 
%
Finally, on network datasets \SQ\ and \QQ\ perform closely since these datasets are characterized by a very limited skeweness, with \QQ\ performing slightly better.


\paragraph*{\textbf{Variable query area.}} 
\label{par:expVQA}

In this batch of experiments, whose results are shown in Figure \ref{fig:plot_varying_qarea},  
we vary the query area. All the queries are equally sized during a single experiment, while the amount of objects is fixed (700K for uniform, 500K for gaussian and network), as well as the query rate ($100\%$) and the number of hotspots (25) for the gaussian datasets.
With uniform distributions \SQ\ and \QQ\ again perform similarly. With gaussian distributions, \SQ\ and \QQ\ maintain consistent performances, even though the advantage of \QQ\ over \SQ\ still holds. Finally, with network datasets \SQ\ and \QQ\ are almost on par, as already observed in the first batch of experiments, with a very slight advantage for \QQ.

\paragraph*{\textbf{Variable amount of objects and variable query area.}} 
\label{par:expVEVQA}

In these experiments we consider different amounts of objects, each one issuing a query 
whose area is decided independently of the other objects and according to
%
a uniform distribution in the  
$[(200u)^2,(400u)^2]$ range.

For this experiments, we again consider uniform, gaussian (25 hotspots) and network datasets. In all the cases considered, the query rate is fixed at 100\%.

\begin{figure}[!h]
    \begin{minipage}{0.5\columnwidth}
      \includegraphics[width=\columnwidth]{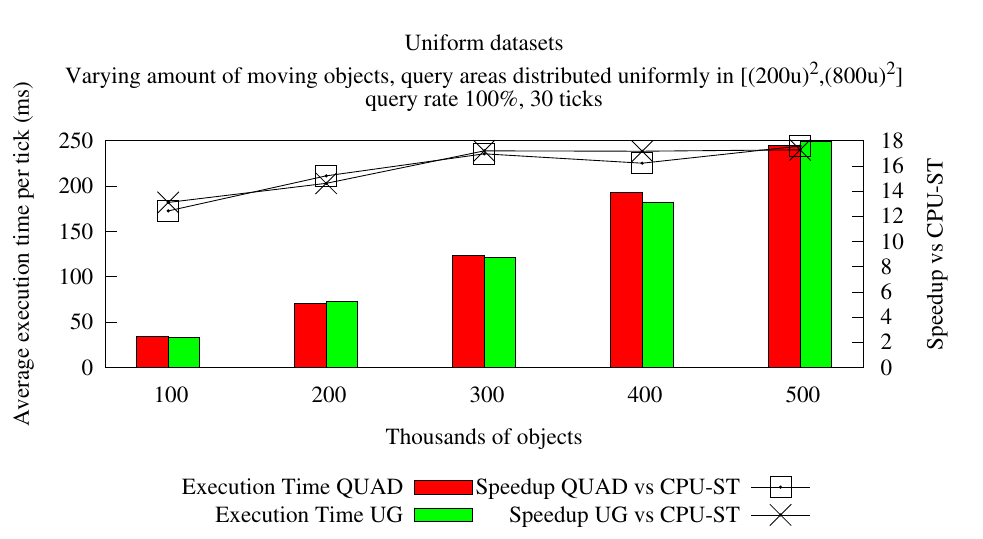}
    \end{minipage}
    \hspace{0.5em}
    \begin{minipage}{0.5\columnwidth}
	\includegraphics[width=\columnwidth]{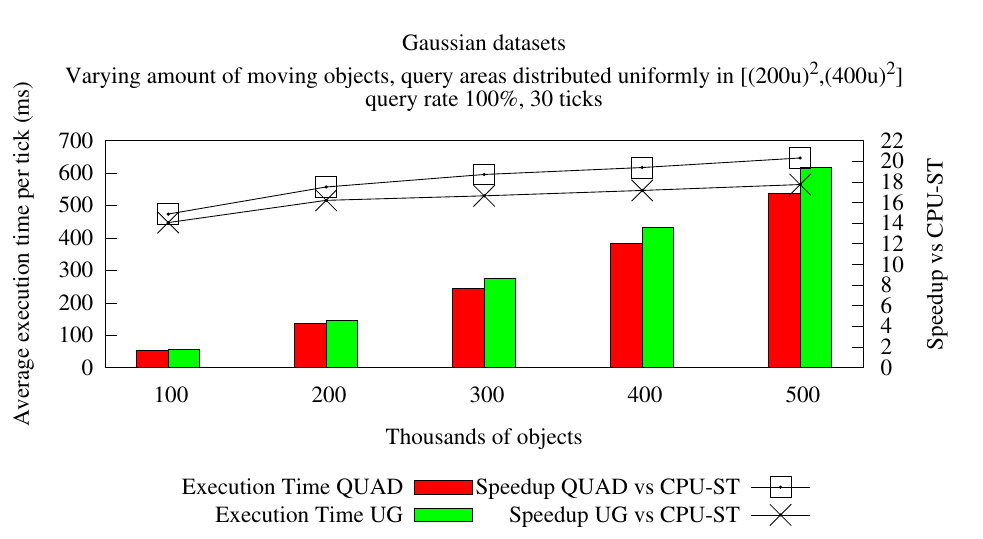}
    \end{minipage}
    \begin{center}
\includegraphics[width=.5\columnwidth]{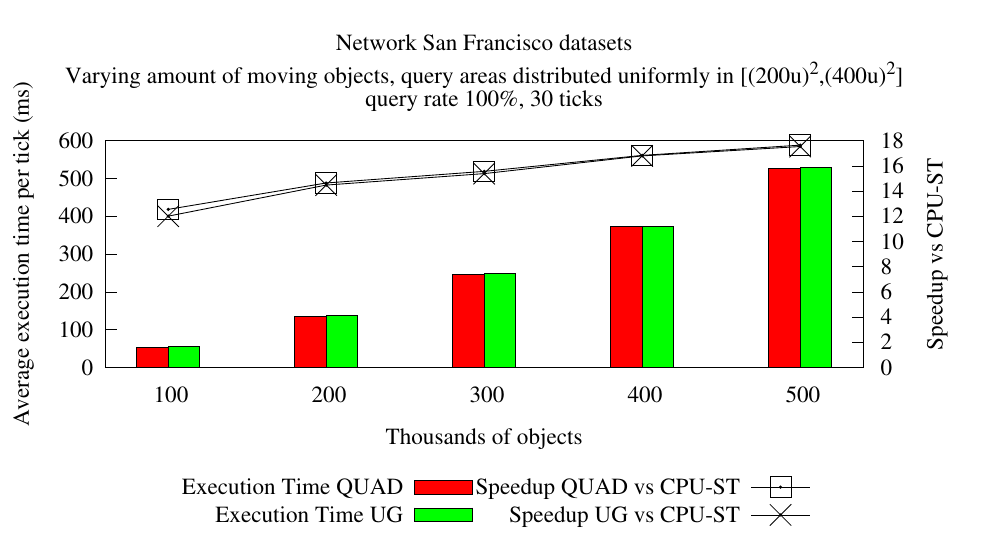}
\end{center}
\caption{Variably sized queries: average running time per tick and speedup versus \CPU.}
\label{fig:plot_network_VEVQA_timings}
\end{figure}

The average execution times per tick measured in  those tests are summarized in Figure \ref{fig:plot_network_VEVQA_timings}. We observe that for that in the first and last case the two algorithms perform likewise, whereas \QQ{} outperform \SQ{} in the second case. Thus, even when the same dataset contains queries having different and uniformly distributed area, we can observe the same trends we highlighted in Study S7 for Figure \ref{fig:plot_varying_objects}.

\subsection{Bandwidth analysis (S8)}

From lemma \ref{lemma:timely2} in Section \ref{sec: qos considerations} we have that the system bandwidth $\beta$, expressed as the amount of queries processed per time unit (indeed, we use the second), is one of the crucial parameters in order to determine a suitable tick duration $\Delta t$, along with a given latency requirement $\lambda$ and a maximum amount of queries which may occur during $\Delta t$, $Q_{max}$. Since $\lambda$ and $Q_{max}$ are fixed, the crucial parameter becomes $\beta$.

Consequently, the goal of this study is to observe how the bandwidth $\beta$ of a given system reacts to a set of dominant factors, such as the \emph{amount} of moving objects, the \emph{query rate} (i.e., the factor of moving objects issuing a query during a time unit), the \emph{query area}, and the \emph{skewness}. \QQ\ will be used to conduct all the experiments.

Figure \ref{fig:plot_bandwidth_VE_VS} presents the results of the first batch of experiments, where we test the behaviour of $\beta$ with respect to different amounts of objects and degrees of skewness. In order to conduct these experiments a set of gaussian datasets were considered. From the Figure we see how the system bandwidth decreases whenever the amount of objects or skewness degree (ranging from uniform-like distributions - 10000 hotspots - to moderately skewed ones - 25 hotspots) increase, due to an increase in the overall amount of containment tests and results that 
the underlying system must handle in the same time unit. We also observe how highly skewed datasets produce
the most notable negative consequences on the performance, thus requiring particular care.

\begin{figure}[!h]
\begin{minipage}{0.5\columnwidth}
\includegraphics[width=\columnwidth]{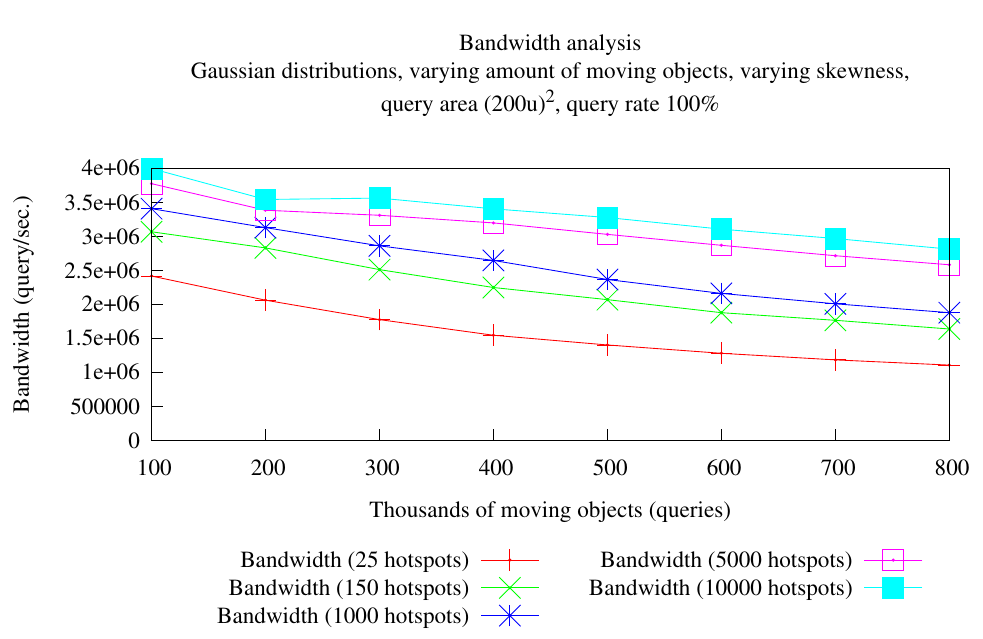}
\caption{System bandwidth analysis when varying the amount of moving objects or the dataset skewness. 
}
\label{fig:plot_bandwidth_VE_VS}
\end{minipage}
\hspace{1em}
\begin{minipage}{0.5\columnwidth}
\includegraphics[width=\columnwidth]{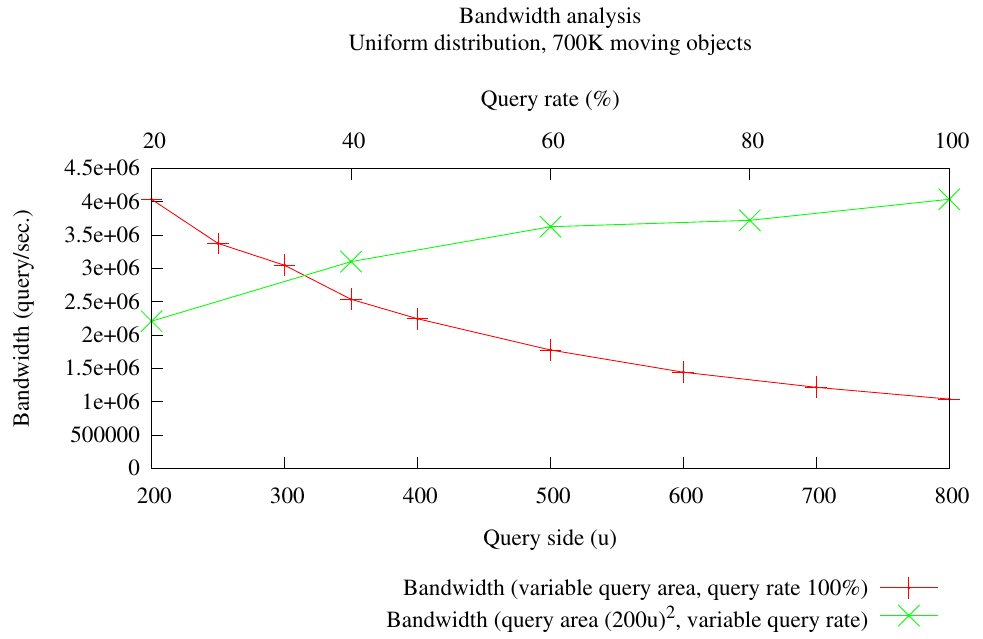}
\caption{System bandwidth analysis when varying the query area or the query rate. 
}
\label{fig:plot_bandwidth_VQA_VQR}
\end{minipage}
\end{figure}

Figure \ref{fig:plot_bandwidth_VQA_VQR} reports the results related to the second batch of experiments, where we analyze the behavior of $\beta$ with respect to the query rate (we observe it corresponds to changing $Q_{max}$) and the query area. 
To this end we consider a set of uniformly distributed datasets characterized by different query rates and query areas.
We observe that increasing the query area decreases the system bandwidth, due to a quadratic increase in the amounts of containment tests and results the system must handle per time unit.
As regards  query rate, we see how the bandwidth increases whenever this parameter is increased. 
Even if this phenomenon may seem counter-intuitive at first, we observe that the action of increasing the query rate has the effect of increasing linearly (and not quadratically) the amount of containment tests and results produced.
These increases, however, are compensated by an increased efficiency of the system. That is, the GPU resources are 
more utilized and thus better exploited, and this in turn increases the overall bandwidth. We note that we observed the same behavior for all of the spatial distribution we considered and for a wide range of choices of parameters.

\section{Related work}
\label{sec:relwork}		

The idea of using the abundant and cheap computational power offered by GPUs in order
to boost spatial joins computations dates back to the era when GPUs did not offer real
general purpose computing capabilities and the use of OpenGL or DirectX APIs were
needed in order to have access to their resources \cite{sun2003hardware}.
The potential of GPUs was clearly understood, but the scenarios, the problems, and
the approaches considered at that time were quite different from those covered in this work.

As pointed out in an extensive review \cite{sowell2013experimental}, the need for
managing continuously incoming and evolving spatial data can be addressed by using simple, light-weight and, in many cases, throwaway data structures.
However, it is crucial that data structures and algorithms contend effectively with skewed data and avoid redundant spatial joins and bad workload distributions as much as possible. In such review the authors claims Synchronous Traversal to be the top performer across several datasets. However, when considering its multi-threaded variant the speedup yielded by Synchronous Traversal (up to $6\times$ with 12 cores) does not follow a linear behaviour as the amount of cores increases, due to inter-thread dependencies and challenges related to finding a proper way to partition the workload among the cores. Indeed, these are serious challenges which we try to tackle in the present work.

Recent studies \cite{vsidlauskas2012parallel, sid11} show how uniform
grid-based solutions are particularly attractive when managing continuously incoming
and evolving spatial data in main-memory multi-core settings. Even if these works do not consider
the architectural peculiarities and limitations of the GPUs, they nonetheless highlight how regular grids,
in general, represent a natural basis for GPU parallelization strategies thanks to their structural regularity \cite{lauterbach2009fast}.

Other works consider the problem of building R-Trees (and possible derivations) from scratch \cite{sowell2013experimental}, even recurring to hybrid approaches based on the combined use of CPU and GPU for range queries computation \cite{luo2012parallel, yu2011parallel}.
While the goals of some of these works are different with respect from the ones of the present work, it is interesting to notice how solving certain problems is particularly recurrent and challenging when processing massive spatial data by using massively parallel architectures, i.e., (i) finding a solution able to distribute the workload in the most uniform way (depending also on the spatial data distribution), (ii) arranging spatial data by using proper GPU-friendly light-weight regular data structures which allow to use the GPUs features effectively, and (iii) exploiting spatial locality as much as possible.

In a recent work~\cite{karras12} in the context of collision detection in computer graphics, the 
the author focuses on extremely fast and efficient GPU-based construction and lookup algorithms for binary radix trees when performing real-time collision detection between 3D objects (thus addressing a similar problem with respect to the one addressed in this work). While the algorithms proposed are able to handle elegantly the skewness possibly characterizing the data, the work doesn't consider the problems of detecting and having to write out huge amounts of results in very short time intervals. The first problem increases remarkably the amount of lookups and traversals in the trees needed to compute a query, while the second problem entails serious issues mainly related to memory throughput maximization and how to avoid memory access contention.
%

%
Previous works~\cite{lauterbach2009fast} already pointed out the advantages of using point-region quadtrees for partitioning a low dimensional space when using the GPUs, thanks to the direct relationship between the quadtrees structural properties and the Morton codes~\cite[Ch.\ 2]{har2011geometric}\cite{raman2008converting}. Indeed, quadtrees fit extremely well the GPUs architectural features, hence allowing to devise fast and efficient algorithms.

We are unaware of existing studies tackling the problem of repeatedly computing sets of range queries over continuously moving objects by means of an hybrid CPU/GPU approach. The most closely related work is focused on point-in-polygon joins \cite{zhang2012pointInPolygon}. 
This work considers scenarios characterized by massive amounts (possibly millions) of static entities, represented by points, and sets of polygons (in the order of few thousands) potentially covering the entities: the goal is to speed up the point-in-polygon tests by exploiting the computational power of GPUs using a novel approach stemming from the traditional \textit{filtering and refinement} schema. This work is similar to the present work in that it exploits point-region quadtrees in order to index the points, and thus improve the workload distribution when determining which point-in-polygon tests have to be computed. However, relevant differences separate the two works: (i) the entities are static, (ii) joins are computed between huge amounts of points and limited sets of polygons (instead of huge sets of range queries issued by the same entities) and (iii) polygons are indexed (through their bounding boxes) by means of a uniform grid and subsequently paired (for the final refinement phase) with sets of potentially overlapped points. This is in turn achieved by indexing each quadtree quadrant minimum bounding rectangle (enclosing the quadrant points) by means of the same uniform grid.
In light of this, we deem that a comparison with our proposals would be not interesting, since the other work tackles different scenarios and consequently does not consider a set of relevant issues having far-reaching consequences, above all the issues related to the continuous management of huge sets of queries and results.

\section{Conclusions}
\label{sec:conclusions}

In this paper we present a novel method, relying on scalable grid-based spatial indices and on ad-hoc data structures, 
capable of computing massive amounts of repeated range queries over continuously moving objects. Since the range 
queries are repeatedly issued by the moving objects, also the queries continuously change their issuing 
location over time. The solution proposed is the first known to exploit GPUs in order to speed-up the query processing and, at the same time, effectively contend with skewed spatial distributions of objects and queries. To achieve these goals, we introduce an hybrid CPU/GPU query processing pipeline. We leverage bitmap-based intermediate data structures as well as quadtree-based spatial index, to exploit effectively the GPUs architectural features.

We extensively test our solution to study its sensitivity to parameters and data distribution. In such experiments we prove several arguments, above all that our solution (i) outperforms a baseline GPU approach thanks to the introduction of lock-free bitmaps used to handle intermediate query results and (ii) achieves a significant performance gain with respect to the best known CPU sequential competitor. We also show that (iii) our proposal is able to outperform an advanced GPU-based uniform grid-based solution, since solutions relying on such indices are unable to fully capture the data skewness - an ability which is needed in order to yield much more uniform workload distributions on GPUs.

As a future direction of research we plan to reuse, at least partly, the hybrid CPU/GPU processing pipeline introduced in this work in order to tackle the computation of repeated \emph{k-nearest neighbours} queries issued by continuously moving objects. Indeed, this class of queries has vast amounts of potential applications, therefore speeding up significantly their processing would represent an important contribution.

\bibliographystyle{unsrt}
\bibliography{gpujoin}

\end{document}